\documentclass[11pt]{article}
\usepackage[margin=1.1in]{geometry}
\usepackage{amsmath,amssymb,amsthm}
\usepackage{graphicx}
\usepackage{float}
\usepackage[hidelinks]{hyperref}
\hypersetup{
  pdftitle={The 11/6 supremum of the Wang-Sitters rounding scheme for graph balancing},
  pdfauthor={Adam Y. Shavit},
  pdfsubject={Approximation algorithms; graph balancing; scheduling on unrelated machines},
  pdfkeywords={graph balancing, restricted assignment, approximation algorithm,
               linear programming relaxation, rounding, Wang-Sitters, makespan}
}

\newtheorem{theorem}{Theorem}
\newtheorem{lemma}[theorem]{Lemma}
\newtheorem{proposition}[theorem]{Proposition}
\newtheorem{corollary}[theorem]{Corollary}
\theoremstyle{definition}

\theoremstyle{remark}
\newtheorem{remark}[theorem]{Remark}
\newtheorem{measurement}[theorem]{Measurement}

\newcommand{\leverid}{arXiv:2609.10004}
\newcommand{\OPT}{\mathrm{OPT}}
\newcommand{\ALG}{\mathrm{ALG}}

\title{The $11/6$ supremum of the Wang--Sitters rounding scheme\\
for graph balancing}
\author{Adam Y. Shavit\\
\small Hunter College and the Graduate Center, CUNY\\
\small \texttt{as1127@hunter.cuny.edu} \quad ORCID 0009-0008-1235-0995}
\date{\today}

\begin{document}
\maketitle

\begin{abstract}
Wang and Sitters' $11/6$-approximation for graph balancing (each job allowed on
at most two machines) is not one algorithm but a set of permitted
executions. Step~1 may return any feasible solution of the fractional
relaxation. Step~3 may match remaining jobs to the rounding's slots in many
ways. We determine exactly what that latitude permits: ratios
arbitrarily close to $11/6$, none reaching it, so $11/6$ is the least
constant bounding every permitted run. We then determine the worst-case
guarantee, against the optimum itself, as a function of the big-job threshold
$\beta$, the fractional share at or above which Step~2 assigns a job. On Wang and Sitters' own range
$\tfrac12 < \beta < 1$ the guarantee is \emph{exactly}
$\max\{\tfrac32 + \tfrac\beta2,\ \tfrac52 - \beta\}$. We extend the same
eligibility rule to $0 < \beta \le \tfrac12$, outside the range they state.
There a big job's two shares can both reach the threshold, so Step~2 acquires a
third choice. We determine it too: \emph{exactly}
$\tfrac32 + (1-\beta)\lfloor 1/\beta\rfloor$, hence unbounded as $\beta$ falls.
The worst-case ratio is therefore known at every threshold in $(0,1)$, and
attained at none. Consequently $2/3$ is the unique optimal threshold, and the
guarantee jumps at one half and again at every threshold $1/k$ below it,
rather than degrading smoothly. A companion note
\cite{Paper2} examines changes that leave the constant unchanged.
\end{abstract}

\section{Introduction}

An approximation algorithm is usually a function: instance in, schedule out. The Wang--Sitters scheme is not. Wang and Sitters' three steps are ``solve
the relaxation'', ``assign the big jobs above a threshold'' and ``round the rest
by a Shmoys--Tardos slot matching'', and two of the three do not determine their
own output. At the least feasible threshold the relaxation's feasible region is
in general not a single point, so two solvers may return different solutions;
and the slot rounding takes \emph{a} perfect matching of jobs to slots, of which
there are typically many. Their $11/6$ bound holds for every such run, which is
what makes it correct. It does not say what the worst run is, and they exhibit
no instance attaining it. This note answers that question, and the answer concerns the relaxation and the rounding rather than Wang and Sitters' particular algorithm: the worst run permitted by the specification approaches $11/6$
against the optimum and no run attains it. The complementary question --- what
the best \emph{matching} buys when Step~1 may return any feasible solution --- belongs to the companion note \cite{Paper2}, which
answers it with the same constant: $11/6$ is the least constant for which
\emph{some} valid matching is always within that factor of the threshold. Optimizing the matching step alone therefore cannot bring the worst-case
constant to $7/4$; that half of the conclusion is established in the companion
note.

\paragraph{Computational status.}
Steps 1--3 as stated run in polynomial time: the linear program is polynomial,
Step~2 is a scan, and a valid matching is found by augmenting paths.
Section~\ref{sec:below} extends the eligibility rule to $\beta \le \tfrac12$,
outside the range Wang and Sitters give and where Step~2 acquires a choice; that
extension is introduced here. The companion
paper \cite{Paper2} instead studies a modified Step~3 that returns a
makespan-\emph{minimizing} valid matching. Finding some valid matching takes
polynomial time, but selecting a makespan-minimizing one for a supplied feasible
fractional solution $x$ is NP-hard, by Theorem~23 of the companion note
\cite{Paper2}, which implements that selection by exhaustive enumeration. That
hardness result takes $x$ as part of the input. It does not settle the jointly
optimized variant, in which the algorithm's own Step~1 chooses the feasible
solution; the complexity of that variant is open.

\begin{figure}[htbp]\centering
\includegraphics[width=0.98\textwidth]{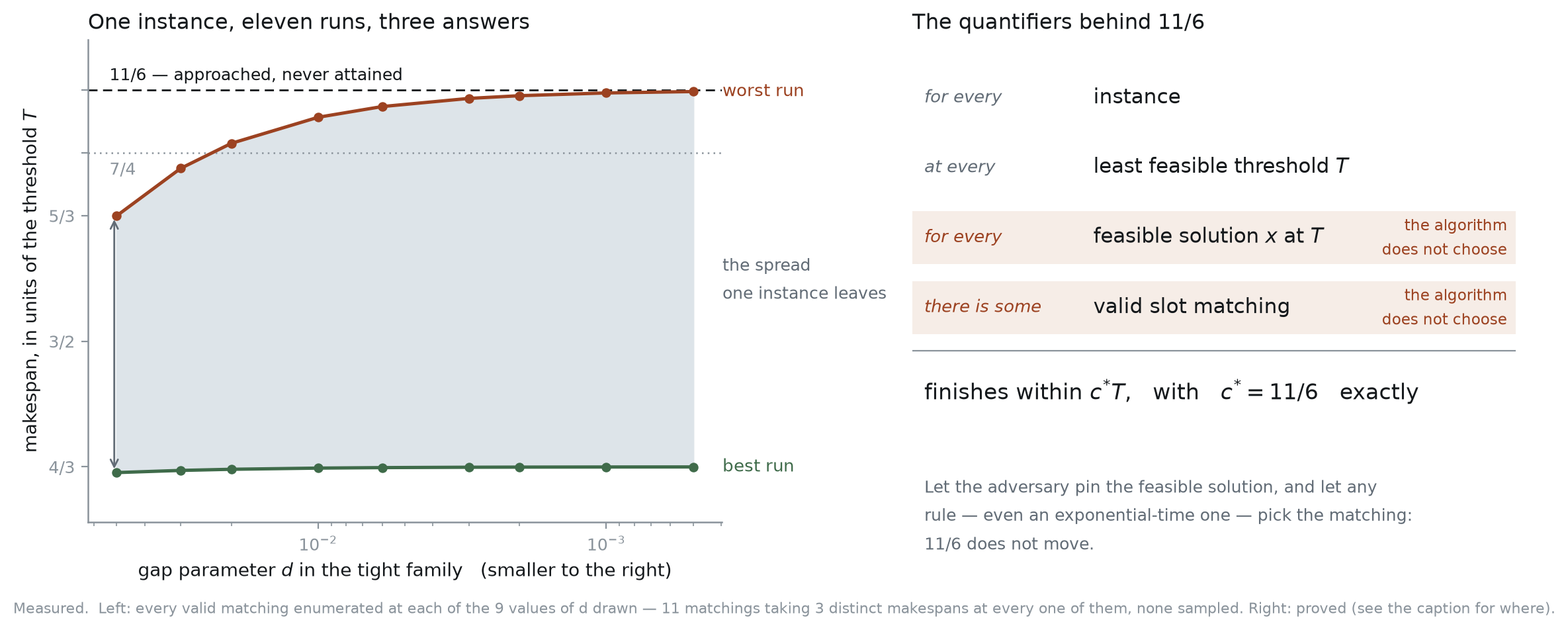}
\caption{The freedom, and the set the constant quantifies over.
\textbf{Left:} the six-job family of Theorem~\ref{thm:tight} at nine values of
the gap $d$. At each one the slot structure is built and \emph{every} valid
matching is enumerated rather than sampled: eleven of them, taking three
distinct makespans. The band is the whole of what the scheme may return on that
one instance, at one threshold, from one feasible relaxation solution --- its
floor tending to $4/3$ and its ceiling climbing toward $11/6$, which
Proposition~\ref{prop:ceiling} shows is never reached. Both ends are limits in
$\delta$: at $\delta = 10^{-3}$ the floor is $8015/6012 = 1.33317$, which
approaches $\tfrac43$ from below as $\delta \downarrow 0$. The band is drawn
against the relaxation target $T$, on a family whose optimum is
$\approx \tfrac43 T$; the three-job family of Theorem~\ref{thm:ratio}, which
carries the same picture against the \emph{optimum}, has eight valid matchings
and a floor at $1 = \OPT$ (Measurement~\ref{meas:eightmatchings}).
\textbf{Right:} the quantifier structure of the threshold-relative constant, proved in the companion note \cite{Paper2}. The two
shaded lines are the choices the scheme does not make: let the adversary pin the
feasible solution, and let any rule, even one that takes the best matching, pick the
matching, and $c^{*} = 11/6$ stays where it is. \textbf{Status: the band is measured}, by
exhaustive enumeration at the nine values of $d$ drawn and at no others; the
constants it approaches, and the right-hand statement, are proved.}
\label{fig:outputisaset}
\end{figure}

In \emph{graph balancing} each job may be assigned to at most two
machines and carries the same size on both. An instance is therefore a
weighted multigraph, a schedule is an orientation of its edges, and the
objective is to minimize the maximum weighted in-degree. Ebenlendr,
Kr\v{c}\'al and Sgall introduced the problem \cite{EKS08} and gave a
polynomial-time $1.75$-approximation \cite{EKS14}; that was still the best
known polynomial-time ratio in 2019 \cite{JR19}, and we know of no later
improvement. Unless P${}={}$NP no
polynomial-time algorithm has ratio strictly smaller than $3/2$
\cite{AJMOZ11}; Ebenlendr, Kr\v{c}\'al and Sgall found the same bound independently \cite{EKS08}. Jansen and Rohwedder \cite[Thm.~1]{JR19} proved
that the configuration linear program has integrality gap at most
$1.749$ for graph balancing. Their argument is a local search not
known to terminate in polynomial time; in their own words it ``does not
give a polynomial time approximation algorithm,'' though it does let the
optimum be estimated within $1.749 + \epsilon$ in polynomial time. The
best \emph{approximation ratio} for graph balancing therefore still
stands at $1.75$, and the present note does not improve it. Graph
balancing is a special case of the restricted assignment problem, and
for that problem the same authors give a quasi-polynomial-time
approximation \cite{JR17ipco} and a configuration linear-program gap
bound of $11/6$ \cite[Theorem~1.1]{JR17soda}, tabulated at
\cite[Table~1]{JR19}; both appear in their extended joint version
\cite[\S4 and Theorem~3.4]{JR20}, which is a different paper from
\cite{JR17soda} despite the shared authors and year.

Wang and Sitters \cite{WS16} observed that a much simpler algorithm
already comes close. Their Section~3 runs to little more than a page:
find a feasible solution of the linear program of \cite{EKS14}; assign
each \emph{big} job (size above half the target) whose fractional share
on some machine reaches $\beta = 2/3$ to that machine; hand the rest to
the Shmoys--Tardos rounding \cite{ST93}. They prove a ratio of
$11/6 \approx 1.833$ and remark that the algorithm and analysis are
``very easy.'' They exhibit no instance showing the bound is attained.
The ratio is worse than \cite{EKS14}'s, and \cite{JR19}, assessing the
algorithm, records its advantage as simplicity.

That analysis bounds every run of what we will call the \emph{Wang--Sitters
rounding scheme}, and the scheme has more than one run. Step~1 asks for \emph{a} feasible relaxation solution and
Step~3 for \emph{any} valid slot matching: each has many, and
nothing in the scheme chooses between them. The scheme is therefore not a
function of its input, and its makespan on a fixed instance is a set of numbers
rather than one of them (Figure~\ref{fig:outputisaset}). This note exercises
that latitude and measures what it costs.

\paragraph{Contributions.}
\begin{enumerate}
\item \textbf{The approximation ratio --- against the optimum --- is tight, and
$11/6$ is a supremum that is never attained} (Section~\ref{sec:tight}). Theorem~\ref{thm:ratio} gives a
three-job, three-machine family with $\OPT = T = 1$ on which a feasible
relaxation solution and a valid matching produce makespan $11/6 - \delta$.
Proposition~\ref{prop:ceiling} proves the matching converse: for every
instance, every feasible relaxation solution and every valid matching, the
makespan is \emph{strictly} below $\tfrac{11}{6}T$. Together these fix
the supremum at $11/6$ and show it is not attained.

\item \textbf{A weaker, relaxation-relative family}
(Section~\ref{sec:lprel}). Theorem~\ref{thm:tight} gives a six-job family
whose relaxation solution is forced up to $O(\delta)$, and on which the
algorithm comes within $O(\delta)$ of $\tfrac{11}{6}T$ without reaching
it, by Proposition~\ref{prop:ceiling}. Its true optimum is
$\tfrac43 + O(\delta)$, so the ratio it forces against the optimum tends to
$\tfrac{11}{8} = 1.375$ and is below it at every positive $\delta$. With the orientation costs deleted, this weaker
phenomenon is already present in both tightness instances of Schwartz and
Yeheskel \cite[Lemma~12]{SY21}, which is where the priority for it lies
(Section~\ref{sec:sy}). The six-job family is also not
$\beta$-independent: for $\beta \le 2/3 - \delta$ the big job is assigned
in Step~2 and the ratio collapses.

\item \textbf{The threshold, over its whole range} (Section~\ref{sec:sy} and
Section~\ref{sec:below}). Proposition~\ref{prop:beta} bounds the guarantee at
threshold $\beta$ from below by $g(\beta) = \max\{\tfrac32 + \tfrac\beta2,\
\tfrac52 - \beta\}$ on $(\tfrac12, 1)$; the matching ceiling is \cite{WS16}'s own
$\beta$-dependent computation, restated with its cases and its strictness in
Proposition~\ref{prop:betaceil}, so the guarantee there is exactly $g(\beta)$
(Corollary~\ref{cor:betaexact}). Above one half the threshold-relative lower bound is \cite{SY21}'s Lemma~12; what is new there is the load-shifting edit that makes their instances bind under Wang and Sitters' non-strict test, and the families measured against the optimum (Theorem~\ref{thm:optexact}). At or below one half the analysis changes
rather than degrades: every big job is seized in Step~2
(Lemma~\ref{lem:allseized}), a machine may be handed $\lfloor 1/\beta\rfloor$ of
such jobs, and Theorem~\ref{thm:belowfamily} constructs instances where
the least threshold at which the relaxation is feasible, written
$T_{\mathrm{LP}}$, already equals the optimum --- so the ratio
$\lfloor 1/\beta\rfloor + \tfrac12$ it reaches is a ratio
against the optimum rather than against an externally chosen feasible target;
Theorem~\ref{thm:belowceil} caps that at
$\tfrac32 + (1-\beta)\lfloor 1/\beta\rfloor$, and a nested family --- one in
which a single machine, the \emph{collector}, is handed every big job the
threshold allows and then has its remaining load budget spent on further slots
behind them --- meets that cap in the limit,
so the guarantee below one half is exactly the ceiling
(Theorem~\ref{thm:belowexact}), approached and never attained. Hence $\tfrac23$
minimizes over all of $(0,1)$ (Corollary~\ref{cor:wholerange}), and
Corollary~\ref{cor:optwholerange} puts the two branches together as a single
function of $\beta$, measured against the optimum: the worst-case ratio of the
scheme is known at every threshold and attained at none.

\item \textbf{Changes that do \emph{not} fix the constant} (companion note
\cite{Paper2}). A Step~3 matching oracle that minimizes makespan leaves the worst-case constant
against the optimum at exactly $11/6$, and restricting Step~1 to a
\emph{vertex} of the relaxation leaves it there too. Both $7/4$ statements about
the oracle are false. Those results, and the reduction and search work behind
them, are stated and proved in \cite{Paper2}; this note cites them and does
not repeat them.
\end{enumerate}

\paragraph{Status of the analytic and computational results.}
Statements labelled Theorem, Proposition, Lemma or Corollary are proved for the
parameter ranges they state. Statements labelled \emph{Measurement} report what
the deposited code returned at stated parameters, and serve as corroboration or
illustration. The threshold results hold throughout their intervals rather than
at sampled points: Proposition~\ref{prop:beta} is proved for every
$\beta \in (\tfrac12, 1)$, each branch carried by a family covering an interval
of thresholds, with Measurement~\ref{meas:sy} as a control on that algebra; and
the results below one half in Section~\ref{sec:below} are proved at every such
threshold, their arithmetic checked in exact rationals rather than floating
point (Measurement~\ref{meas:belowhalf}) because several steps turn on
equalities. Measurement~\ref{meas:ceilingoffbeta}, which counts random instances
breaking the $\tfrac{11}{6}$ ceiling below $\beta = \tfrac12$, is a sweep: its
seven-machine witness is certified in exact rational arithmetic and its
six-machine maximum is a floating-point recomputation (the measurement states
which is which), while the counts around them are a sample. Appendix~\ref{app:restatement} carries a \emph{restatement check} on selected
central results: a re-reading of each statement against its proof, to verify
that the quantifiers it states are the ones the proof establishes. The check is
supplementary rather than exhaustive.

\section{Preliminaries}\label{sec:prelim}

The scheme is run at a \emph{target} $T$ and asks whether the jobs can be
scheduled within it. An instance has at least one job, and every job has a
positive size and at least one allowed machine. Write $p_j$ for the size of job $j$ and $P_T(I)$ for the
set of feasible solutions of the relaxation below, on instance $I$ at target
$T$. The \emph{least feasible target} is
\[
T_{\mathrm{LP}}(I) \ =\ \min\,\{\, T > 0 \ :\ P_T(I) \neq \varnothing \,\} ,
\]
a minimum rather than an infimum. Two facts make it one: feasibility is monotone in $T$, and the infimum is attained.

\emph{Feasibility is monotone.} Raising the target only weakens constraints: the
load bound grows, the support condition frees variables as fewer jobs have
$p_j > T$, and the set of big jobs $\{j : p_j > T/2\}$, which the third
constraint counts, can only shrink. Those two sets change only at the finitely
many numbers $p_j$ and $2p_j$, and on each interval between consecutive ones
they are constant. Each such interval is \emph{closed at its left end}, because
$p_j > T$ and $p_j > T/2$ both fail at the breakpoint itself.

\emph{The infimum is attained.} Write $T^{*}$ for it and take
$T_n \downarrow T^{*}$ with $x^{(n)} \in P_{T_n}(I)$. The $x^{(n)}$ lie in the
compact cube $[0,1]^{E}$, so a subsequence converges to some $x^{*}$, and the
assignment equalities and the load inequalities pass to the limit. Only the
big-job constraint needs more, since it sums over a set that moves with the
target. For all large $n$ the interval $(T^{*}, T_n]$ contains no breakpoint, so
$\{j : p_j > T_n/2\} = \{j : p_j > T^{*}/2\}$: the constraint at $T^{*}$ is the
same constraint, not a stronger one. Hence $x^{*} \in P_{T^{*}}(I)$.

\textbf{This note computes $T_{\mathrm{LP}}$ by binary search on $T$, which approximates it; the minimum above is what defines it.} Every guarantee, supremum and attainment statement here is evaluated at the exact $T_{\mathrm{LP}}$; the recorded computational runs that use the binary search carry its bisection tolerance. The distinction matters because the set
of big jobs --- and hence the relaxation itself --- changes when $T$ crosses a
breakpoint $2p_j$.

$T_{\mathrm{LP}}$ is at most the true optimum $\OPT$, because an integral
schedule of makespan $\OPT$ is itself a feasible fractional solution at $\OPT$. \textbf{$T_{\mathrm{LP}}$ and $\OPT$ need not be equal, and most of this note turns on keeping them apart.} Throughout, $T$ is
normalized to $1$.

The relaxation is that of \cite{EKS14}, also used by \cite{WS16}. The
variable $x_{ij} \in [0,1]$ is the fraction of job $j$ placed on machine $i$,
and is supported on the at most two machines job $j$ is allowed to use:
\[
\sum_i x_{ij} = 1 \ \ \forall j, \qquad
\sum_j x_{ij} p_j \le 1 \ \ \forall i, \qquad
\sum_{j : p_j > 1/2} x_{ij} \le 1 \ \ \forall i ,
\]
together with $x_{ij} = 0$ whenever $p_j > 1$. The third family --- at
most one big job's worth of fraction per machine --- is what separates
this relaxation from the assignment relaxation. It is \emph{not} the
configuration linear program, which is a strictly stronger relaxation with a
variable per machine and configuration, and nothing in this note is a claim about rounding
that one. The distinction matters twice over: the relaxation displayed above
has integrality gap exactly $7/4$, as recorded below, while the
bounds on the configuration linear program that this note cites for priority are
separate results about a different program. Job sizes are
\emph{positive}. That is standard for the problem, since an instance is a
weighted multigraph, but we state it because one step needs it:
Lemma~\ref{lem:chain}'s strictness --- that a machine's last slot carries
positive weighted mass --- fails on a machine whose
last slot carries only jobs of size zero, and that strictness is what
case~(S) of Proposition~\ref{prop:ceiling} uses to be strict rather than
weak. The relaxation's integrality gap is
exactly $7/4$: the upper bound is the algorithm of \cite{EKS14} and the
lower bound their three-path family, three long odd paths joining two
vertices with weights alternating $1$ and $1/2 - \varepsilon$ and a
dedicated load $1/4$ everywhere.

\paragraph{The scheme, and what an execution is.} All three steps below are
Wang and Sitters' \cite[\S3]{WS16}, not \cite{EKS14}'s; the two are different
algorithms with different ratios, and everything in this note is about theirs.
Call a job \emph{big} when its size exceeds $T/2$.

Because two of the three steps do not determine their own output, we fix the object of
study before stating anything about it. Write $P(I)$ for the set of feasible
solutions of the relaxation on instance $I$ at threshold $T$, and, for
$x \in P(I)$, write $\mathcal{M}(I, x)$ for the set of valid slot matchings
available in Step~3 after Step~2 has run on $x$. The \emph{execution space} of
the scheme on $I$ is
\[
  \mathcal{E}(I) \;=\; \{\,(x, M) \;:\; x \in P(I),\ M \in \mathcal{M}(I, x)\,\},
\]
and $\ALG(I, x, M)$ is the makespan the scheme produces on that execution. Two
coordinates suffice at the published threshold and at every $\beta > \tfrac12$;
at or below one half Step~2 acquires a choice of its own and the triple
$(x, A, M)$ replaces the pair, as Section~\ref{sec:below} sets out. The
quantity this note computes is
\[
  R_{\mathrm{WS}} \;=\; \sup_{I}\ \sup_{(x,M)\,\in\,\mathcal{E}_{T_{\mathrm{LP}}(I)}(I)}
  \frac{\ALG(I, x, M)}{\OPT(I)} ,
\]
where the subscript records that the execution space is evaluated at the least
feasible target, which is where the scheme is run: the numerator is a makespan
at that target and the denominator is the true optimum. $\mathcal{E}(I)$ is written without the subscript below only where the
target is fixed by the surrounding text.
We distinguish three objects throughout, and refer to them by these names.

\smallskip
\centerline{\begin{tabular}{lll}
\hline
Name & Step~1 returns & Step~3 returns\\
\hline
the \emph{scheme} (this note's object) & any $x \in P(I)$ & any $M \in \mathcal{M}(I,x)$\\
an \emph{implementation} & one $x$, solver-determined & one $M$, routine-determined\\
the \emph{oracle} & any $x \in P(I)$ & a makespan-minimizing $M$\\
\hline
\end{tabular}}
\smallskip

\noindent
Every theorem below is about the first or the third. No theorem below characterizes a particular implementation: an implementation realizes one execution per instance, and
which one depends on a solver and a matching routine that the published
specification does not fix. A report of what a solver actually returned is
evidence about one implementation, and this note labels it as a measurement,
never as a bound.

\paragraph{The steps.}
\begin{description}
\item[\textbf{Step 1.}] Return \emph{a} feasible solution $x$ of the
relaxation above.
\item[\textbf{Step 2.}] Assign big job $j$ to machine $i$ whenever
$x_{ij} \ge \beta$, with $\beta = 2/3$. Such a job is scheduled on $i$ and
takes no further part: delete its fractions from $x$ on both of its machines.
No machine receives two, because the relaxation's third constraint gives each
machine at most one big job's worth of fraction and $\beta > 1/2$ --- a step
that is exactly as strong as its hypothesis, and Section~\ref{sec:below} is
about what the scheme does once the hypothesis goes.
\item[\textbf{Step 3.}] Build the \emph{slot structure} of \cite{ST93},
defined immediately below, from what is left of $x$, and return \emph{any}
valid matching of the surviving jobs to slots (Figure~\ref{fig:freedoms}).
\end{description}
Steps~1 and~3 each return \emph{a} feasible object rather than a uniquely
specified one. Wang and Sitters prove that every run finishes within
$\tfrac{11}{6}T$.

\begin{figure}[htbp]\centering
\includegraphics[width=0.95\textwidth]{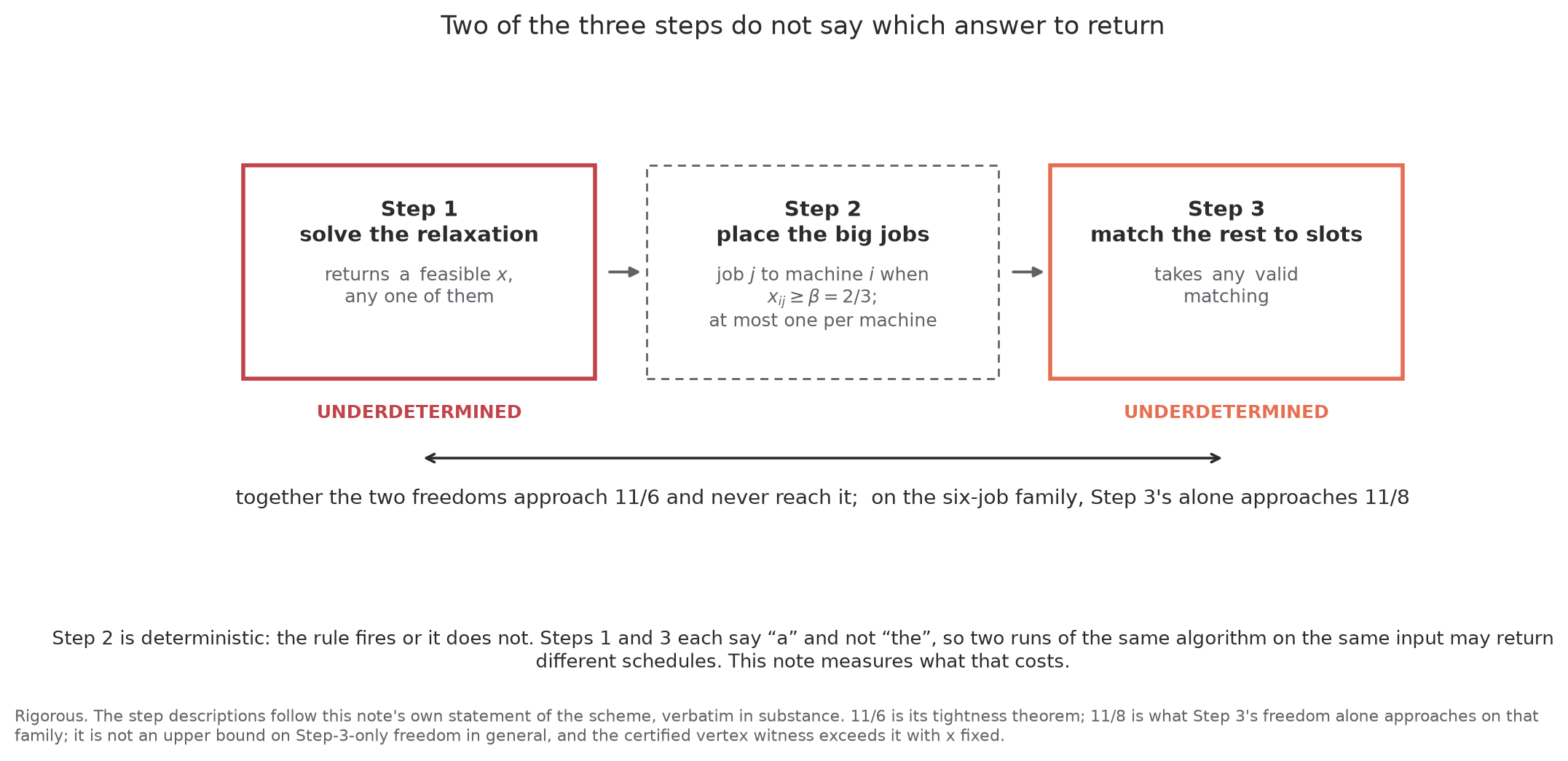}
\caption{The scheme in three steps, with the two that do not determine their own
answer marked. At the published threshold $\beta = 2/3$, Step~2 is a rule that
either fires or does not; below one half it acquires a choice of its own, which
Section~\ref{sec:below} treats.
Steps~1 and~3 each return \emph{a} feasible object rather than \emph{the}
one, so two runs on the same input may return different schedules. The
note quantifies that cost: $11/6$ when both steps are free
(Theorem~\ref{thm:ratio}), and a ratio tending to $11/8$ on the family that
isolates Step~3's latitude (Theorem~\ref{thm:tight}).}
\label{fig:freedoms}
\end{figure}

The slot structure is built as follows: on each machine,
order its fractional jobs by nonincreasing size and cut the fractional
mass into unit slots, the last one possibly shorter (in the per-machine
notation below its length is $M_i - K_i + 1$, which is $1$ exactly when the
mass $M_i$ is a whole number); a job is adjacent
to every slot its own fraction touches. Slots and job fractions are
half-open intervals $[\,\cdot\,,\cdot\,)$, and \emph{touches} means
nonempty intersection; a job whose fraction ends exactly where a slot
begins is not adjacent to that slot (Figure~\ref{fig:slots}). Proposition~\ref{prop:ceiling} below is false
under the closed reading, so we fix the convention here.

\begin{figure}[htbp]\centering
\includegraphics[width=0.95\textwidth]{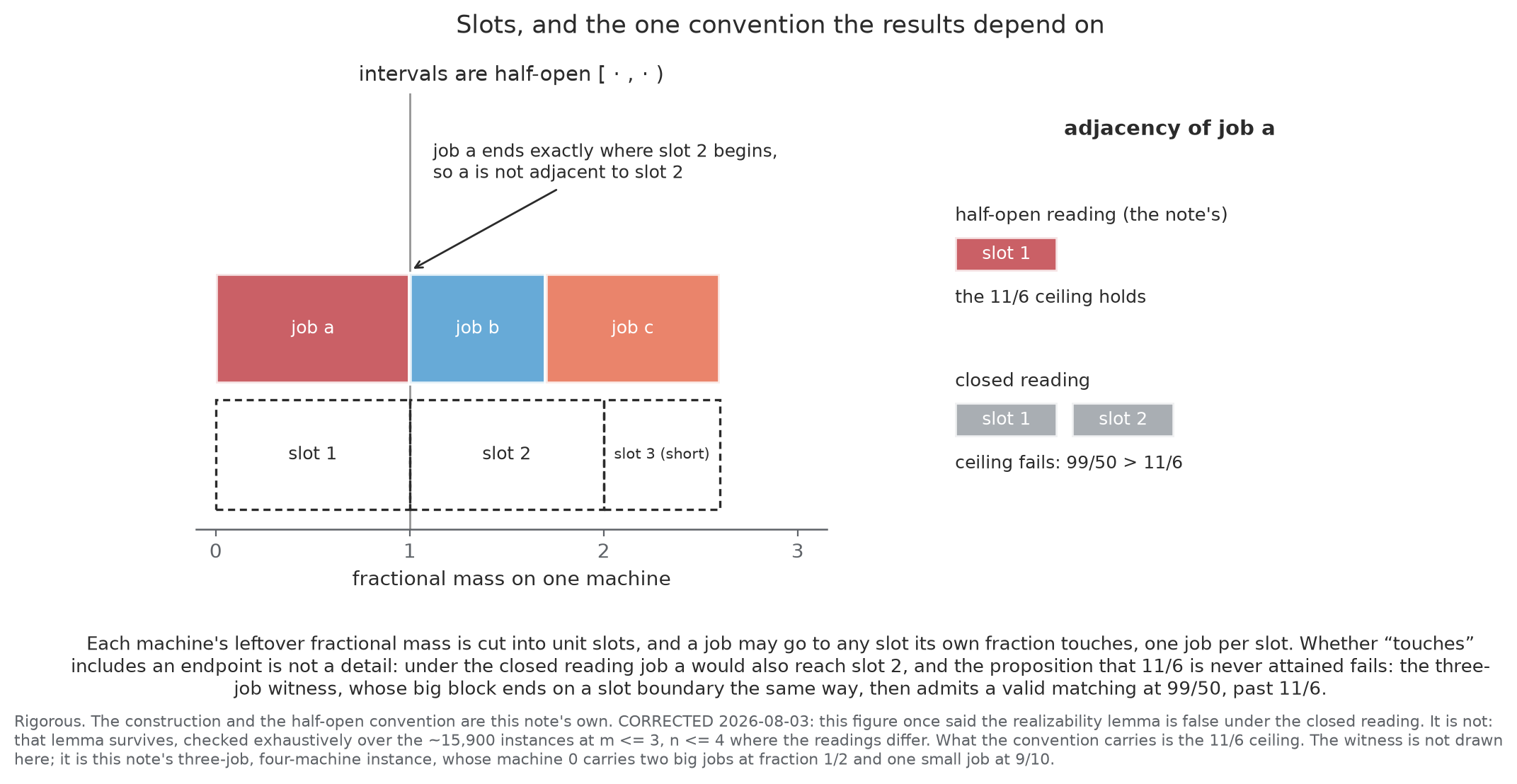}
\caption{The slot structure, and the half-open convention Section~\ref{sec:prelim} fixes. Left: a
machine's leftover fractional mass, cut into unit slots, with three job
fractions laid across it; job~$a$ ends exactly where slot~2 begins, so
under the half-open reading it is not adjacent to slot~2. Right: the two
readings disagree about that single edge, and the disagreement is not
cosmetic: Proposition~\ref{prop:ceiling} is false under the closed reading,
though an exhaustive exact-rational check over the small instances where the
two readings can differ ($m \le 3$, $n \le 4$, fraction denominators at most
$4$) finds no counterexample to the realizability lemma of \cite{Paper2} in that
finite test class.}
\label{fig:slots}
\end{figure}

\emph{It is also \cite{ST93}'s own convention, not a choice of ours, and the
strictness is visible in their construction.} Building the bipartite graph
$B(x)$ they set $k_i = \lceil \sum_j x_{ij} \rceil$, order the fractions by
nonincreasing $p_{ij}$, fill slot $v_{i1}$ to a total of exactly one and
continue with the smaller jobs, so that only the \emph{last} slot can be short, and none is when the
mass is a whole number;
and at each boundary they add the edge $(v_{i,s+1}, w_{j_s})$ carrying a job
into the next slot only ``if $\sum_{j \le j_s} x_{ij} > s$'' \cite[p.~464]{ST93},
a \emph{strict} inequality. A job whose fraction ends exactly at a slot
boundary therefore receives no edge to the following slot, which is precisely
the half-open reading. Two consequences matter here: the closed reading is not an alternative rendering of \cite{ST93} but a departure from it, and $\lceil\,\cdot\,\rceil$ gives a machine of fractional mass exactly $2$ exactly two slots, a count on which the pigeonhole argument of \cite{Paper2} depends. A \emph{valid matching}
assigns each remaining job to one adjacent slot, at most one job per slot;
one exists by Hall's theorem.
The half-open convention matters for one result below and not for another.
It is \emph{not} needed for the realizability lemma of \cite{Paper2} \emph{as measured}, though we do not prove that: the
lemma's proof uses positive-length contact between a job's interval and a slot,
which half-openness guarantees and the closed reading does not, so the
proof as written does not transfer. The claim rests on
an exhaustive exact-rational check over
the instances where the two readings can differ at all---those with a job
interval ending exactly on a slot boundary---which found every job--slot edge
realizable in both, roughly $15{,}900$ such instances at $m \le 3$,
$n \le 4$ (\texttt{p29amb\_convention\_probes/}). The convention is needed for
Proposition~\ref{prop:ceiling}. Under the closed reading that
proposition is false, and the case that breaks it is small enough to state
in full: three jobs on four machines at $T = 1$, two big jobs of size
$99/100$ split $1/2$--$1/2$ from machine~0 onto private partners, and one
job of size $1/200$ taking $9/10$ of its share on machine~0. Machine~0's
big-job share is then exactly $1$, so its big block ends precisely on the
first slot boundary; no share reaches $\beta = 2/3$, so Step~2 assigns nothing. Half-open there are six valid
matchings, the worst finishing at $199/200$. The closed reading adds two
edges, and it is the second of them that matters: because machine~0's big
block ends exactly on the slot boundary, the second big job becomes adjacent
to slot~2. That is precisely the step case~(N1) of
Proposition~\ref{prop:ceiling} takes when it passes from \emph{all big fraction
lies in slot~1} to $s_z \le \tfrac12$ for $z \ge 2$. The fraction does still lie
in slot~1 under either reading; what the closed reading adds is that a job whose
fraction \emph{ends} at the boundary counts as adjacent to slot~2 anyway, and
that job has size $\tfrac{99}{100}$. There are then ten valid matchings, and the one that puts both big
jobs on machine~0 finishes at $99/50 = 1.98 > 11/6$. The small job's own extra
edge --- to slot~1 --- changes nothing on its own: adding it alone leaves
every matching at $199/200$. So the half-open reading is what makes the ceiling
hold. It is drawn in Figure~\ref{fig:halfopen}.

\begin{figure}[htbp]\centering
\includegraphics[width=0.95\textwidth]{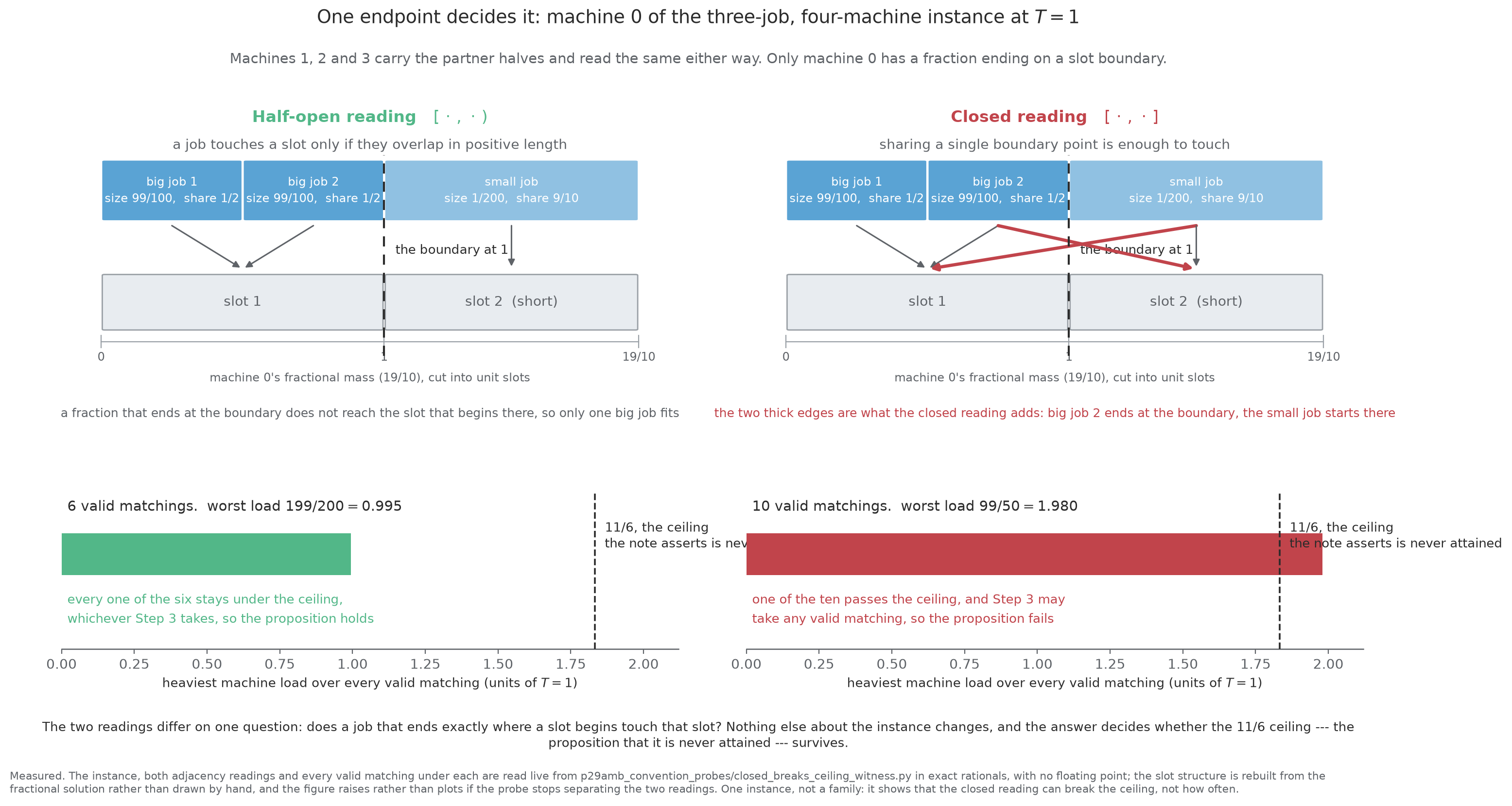}
\caption{The three-job counterexample of Section~\ref{sec:prelim}, drawn under
both readings, showing why the ceiling of Proposition~\ref{prop:ceiling} needs the
half-open convention.
Machine~$0$'s two big fractions end exactly at the first slot boundary. Under
the half-open reading that edge is absent; under the closed reading it is present,
the second big job becomes adjacent to slot~2, and both big jobs can be placed on
machine~$0$. Six valid matchings with worst makespan $199/200$ become ten
with worst $99/50 > \tfrac{11}{6}$. \textbf{Status: measured} by exhaustive
enumeration over this one instance, which shows the closed reading
\emph{can} break the ceiling, not how often.}
\label{fig:halfopen}
\end{figure}

\paragraph{Per-machine notation.}
Fix a machine $i$ and a valid matching. Write $b_i$ for the size of the
job assigned to $i$ in Step~2 (zero if none); $L'_i$ for the relaxation load
on $i$ of the jobs that entered the slot structure. Then
$L'_i \le 1 - \beta b_i$: machine $i$'s load constraint gives
$\sum_j x_{ij} p_j \le 1$, and the job Step~2 took contributed at least
$\beta b_i$ to that sum, since its share reached $\beta$. Write
$M_i$ for the total fraction on $i$ of those same jobs and $K_i = \lceil M_i \rceil$ for
the number of slots; $s_z$ for the largest size adjacent to slot $z$; and
$w_z$ for the size-weighted mass of slot $z$, so that
$\sum_{z \le K_i} w_z = L'_i$. Recall that a job is \emph{big} when its size
exceeds $T/2$, which is $1/2$ here.

\begin{lemma}[chain bound]\label{lem:chain}
For every machine $i$ carrying at least one slot job and every valid
matching:
\begin{enumerate}
\item[(i)] $w_{K_i} > 0$, and $s_z \le w_{z-1}$ for $2 \le z \le K_i$,
hence $\sum_{z=2}^{K_i} s_z \le L'_i - w_{K_i}$;
\item[(ii)] the load of $i$ is at most $b_i + s_1 + L'_i - w_{K_i}$.
\end{enumerate}
\end{lemma}

\begin{proof}
Both inequalities are \cite{WS16}'s. Their equation~(1) reads: ``Since jobs
are ordered by $\prec_i$, the job matched to $v^i_z$ has processing time at
most $p^i_z \le L^i_{z-1}$ for all $z \ge 2$'', which is the inequality
$s_z \le w_{z-1}$ of~(i); the display immediately following it bounds the
total load of machine $i$ by $p^i_1 + \sum_{z=1}^{y_i-1} L^i_z$, where their
$y_i$ is our slot count $K_i$ and their $L^i_z$ our $w_z$, which
is~(ii) without the final term. Our contribution is only the strictness
$w_{K_i} > 0$, and the consequent subtraction of $w_{K_i}$, which is what
later gives the strict inequality of Proposition~\ref{prop:ceiling}. We give
the argument in the present notation for self-containedness.

Slot $K_i$ carries fractional mass $M_i - (K_i - 1) > 0$ of jobs of
positive size, so $w_{K_i} > 0$; this is the one place the positivity fixed
in Section~\ref{sec:prelim} is used. For $2 \le z \le K_i$, slot $z-1$ is
full, i.e.\ has length exactly $1$. Let $j^\star$ be the first job in
the machine's nonincreasing-size order that touches slot $z$, so
$p_{j^\star} = s_z$. Every job occupying mass in slot $z-1$ precedes
$j^\star$ in that order or equals it, hence has size at least $s_z$;
integrating over the unit length of slot $z-1$ gives
$w_{z-1} \ge s_z$. Summing over $z = 2, \dots, K_i$ and using
$\sum_{z \le K_i - 1} w_z = L'_i - w_{K_i}$ gives (i). For (ii), the job
matched to slot $z$ has size at most $s_z$, and at most one job is
matched per slot, so the load is at most
$b_i + \sum_{z \le K_i} s_z \le b_i + s_1 + L'_i - w_{K_i}$.
\end{proof}

\section{The approximation ratio is tight at \texorpdfstring{$11/6$}{11/6}}\label{sec:tight}

\begin{theorem}\label{thm:ratio}
For every $\delta \in (0, 1/6)$ there is an instance on three machines
with three jobs, whose optimum is $\OPT = 1$ and whose relaxation is
feasible at $T = 1$ and at no smaller $T$, together with a feasible
solution of the relaxation at $T = 1$ and a valid slot matching, such
that Steps 2--3 with $\beta = 2/3$ produce makespan
$\tfrac{11}{6} - \delta$. Hence
\[
\sup \ \frac{\ALG}{\OPT} \ \ge \ \frac{11}{6},
\]
the supremum being over instances, over feasible solutions returned by
Step~1, and over valid matchings taken by Step~3.
\end{theorem}

\begin{proof}
Machines $\{0,1,2\}$. Jobs: $q$ of size $1$ allowed on $\{0,1\}$; $h$ of
size $1/2$ allowed on $\{0,2\}$; $r$ of size $1/3 - \delta$ allowed on
$\{0\}$ only.

\emph{The optimum.} The schedule $q \mapsto 1$, $h \mapsto 2$,
$r \mapsto 0$ has loads $1/3 - \delta$, $1$ and $1/2$ on machines $0$, $1$
and $2$, so $\OPT \le 1$;
and $\OPT \ge p_q = 1$. Hence $\OPT = 1$. Since $p_q = 1$ the relaxation
is infeasible for any $T < 1$, so the binary search returns $T = 1$.

\emph{A feasible relaxation solution.} Take
\[
x_{q0} = \tfrac13 + \delta, \quad x_{q1} = \tfrac23 - \delta, \quad
x_{h0} = \tfrac23, \quad x_{h2} = \tfrac13, \quad x_{r0} = 1 .
\]
Machine $0$'s load is
$(\tfrac13 + \delta) \cdot 1 + \tfrac23 \cdot \tfrac12 +
1 \cdot (\tfrac13 - \delta) = 1 \le T$; machine $1$'s is
$\tfrac23 - \delta$; machine $2$'s is $\tfrac16$. The only big job is
$q$, and $x_{q0}, x_{q1} \le 1$, so the big-job constraints hold. The
point is feasible.

\emph{Step~2 assigns no job.} Both of $q$'s shares, $\tfrac13 + \delta$ and
$\tfrac23 - \delta$, are strictly below $\beta = 2/3$, so no job is
assigned in Step~2 and all three enter the slot structure.

\emph{The slot structure and a bad matching.} Machine $0$ carries
fractional mass
$(\tfrac13 + \delta) + \tfrac23 + 1 = 2 + \delta > 2$, hence three
slots. Sorted by nonincreasing size, $q$ occupies
$[0, \tfrac13 + \delta)$, $h$ occupies $[\tfrac13 + \delta, 1 + \delta)$
and $r$ occupies $[1 + \delta, 2 + \delta)$. So $q$ touches slot $1$;
$h$ touches slots $1$ and $2$; $r$ touches slots $2$ and $3$. Matching
$q$ to slot $1$, $h$ to slot $2$ and $r$ to slot $3$ uses each slot at
most once and respects adjacency, so it is a valid matching, and it puts
all three jobs on machine $0$, for load
$1 + \tfrac12 + \tfrac13 - \delta = \tfrac{11}{6} - \delta$.

Since $\OPT = 1$, the ratio realized is $\tfrac{11}{6} - \delta$; letting
$\delta \downarrow 0$ gives the supremum claim.
\end{proof}

\begin{remark}[what the family does and does not require]
The family exercises the freedom in Step~1 as well as the freedom in
Step~3. The exhibited $x$ is not a vertex of the relaxation polytope, and the
reason is a segment through it rather than a count of tight constraints ---
the assignment equalities are tight at \emph{every} feasible point, so counting
them settles nothing. Moving $(x_{q0}, x_{h0})$ by $(t, -2t)$ keeps the point
feasible for $|t|$ small, in both directions, which is what a vertex forbids: the perturbed point satisfies the same equalities and leaves
every load and big-job budget within its bound, which is one line of
arithmetic. Step~2 removes $q$ at both endpoints of that segment: at
$x_{h0} = 1$ one has $x_{q1} = 5/6 - \delta \ge \beta$, and at
$x_{h0} = 0$ one has $x_{q0} = 2/3 + \delta \ge \beta$, the other way. The bad region is the
open sliver $x_{q0} \in (1/3,\, 1/3 + 2\delta)$, of width $2\delta$. So
the theorem is a statement about the scheme as Wang and Sitters state
it (``find \emph{a} feasible solution''), and it would not follow for a
variant that pinned Step~1 to a vertex. The supremum is nonetheless
$\tfrac{11}{6}$ there as well: the companion note \cite{Paper2} exhibits, for
every $g \ge 2$, a point that \emph{is} a vertex of the relaxation and at which
every valid matching finishes at $\tfrac{11}{6} - \tfrac1{4g}$. So a solver
returning a basic feasible solution --- which is what a simplex implementation
does, and what makes this family's own $x$ unreachable in practice --- does not
escape the constant; it only escapes this witness.
\end{remark}

\begin{measurement}\label{meas:eightmatchings}
On the family of Theorem~\ref{thm:ratio}, the exhaustive enumeration of
all valid matchings (there are $8$) gives worst-case makespan exactly
$\tfrac{11}{6} - \delta$ and best-case makespan exactly $1 = \OPT$, for
$\delta \in \{10^{-2}, 10^{-3}, 10^{-4}\}$: measured
$1.8233333$, $1.8323333$, $1.8332333$. Solving Step~1 with GLOP instead --- a run reported here and not
regenerated by the file this measurement names ---
returns an integral $x$ on this instance and makespan $1$, which is why
the adversarial point must be supplied explicitly.
\texttt{p29amb\_ratio\_family\_check.py} regenerates this row in exact rational
arithmetic and writes \texttt{p29amb\_ratio\_family\_check.json}.
\end{measurement}

\begin{proposition}[$11/6$ is never attained]\label{prop:ceiling}
Let the relaxation be feasible at $T$, let $x$ be any feasible solution,
and let any valid matching be taken in Step~3 with $\beta = 2/3$. Then
every machine's load is \emph{strictly} less than $\tfrac{11}{6}T$.
Consequently the supremum in Theorem~\ref{thm:ratio} equals $11/6$ and is
not attained.
\end{proposition}

\noindent
Passing from the bound to the supremum requires a condition on $T$: the theorem's ratio is $\ALG/\OPT$ while the
proposition bounds $\ALG/T$, and the two agree only when $T \le \OPT$. That
holds here because $T$ is the \emph{least} threshold at which the relaxation
is feasible, and any integral schedule of makespan $\OPT$ is in particular a
feasible fractional solution at $\OPT$. Without it the conclusion is false
rather than merely unproved: an instance run at a threshold above its
optimum can push $\ALG/\OPT$ well past $11/6$, and
Theorem~\ref{thm:tight} exhibits an instance with $\OPT > T$, so the
distinction is live in this note and not a technicality.

\begin{proof}
Normalize $T = 1$ and fix a machine $i$. If $i$ carries no slot job its
load is $b_i \le 1$. Otherwise Lemma~\ref{lem:chain}(i) gives
$w_{K_i} > 0$, and we distinguish three cases. They are exhaustive and
disjoint for the trivial reason: either $i$ holds a Step-2 job or it does
not, and if it does not, the job matched to slot~$1$ is either big or is not.
The split itself needs no justification; case~(S)'s use of it does, and the
justification is \cite{WS16}'s observation that a machine holding a Step-2 job
has no unassigned big fraction at all (an unassigned big job has both shares
in $(1/3, 2/3)$, while a Step-2 job already consumes at least $2/3$ of the
machine's big-job budget).

\emph{(N0) No Step-2 job, and the slot-$1$ occupant is small or absent.}
Here $b_i = 0$, and the case hypothesis says that the job \emph{matched}
to slot $1$ has size at most $1/2$, or that no job is matched there. It
does not say $s_1 \le 1/2$: $s_1$ is the largest size \emph{adjacent} to
slot $1$, and a big job adjacent to slot $1$ may be matched on its other
machine. So we stop one step short of Lemma~\ref{lem:chain}(ii) and use
its proof directly. The load of $i$ is the sum over slots of the size
matched there; slot $1$ contributes at most $1/2$ by the case hypothesis,
and slot $z$ contributes at most $s_z$ for $z \ge 2$. Bounding the latter
sum by Lemma~\ref{lem:chain}(i),
\[
\text{load} \ \le\ \tfrac12 + \sum_{z=2}^{K_i} s_z
\ \le\ \tfrac12 + L'_i - w_{K_i}
\ <\ \tfrac12 + 1 \ =\ \tfrac32 ,
\]
using $L'_i \le 1$ and $w_{K_i} > 0$. For $K_i = 1$ the middle sum is
empty and $L'_i - w_{K_i} = 0$, so the bound holds there as well.

\emph{(N1) No Step-2 job, and a big job $q$ is matched to slot $1$.}
Write $x_q = x_{iq}$. If $K_i = 1$ the load is $p_q \le 1 < 11/6$, so
assume $K_i \ge 2$. All big fraction lies in slot $1$, because the
big-job budget is $1$ and bigs sort first; so $s_z \le 1/2$ for
$z \ge 2$. This step is \cite[Observation~3]{SY21}, whose statement is
exactly it: ``Let $e$ be an edge in $\mathrm{slot}(u,i)$ such that
$i > 1$. Then $p_e \le 1/2$,'' proved by the same
sorting-and-budget argument. Slot $1$ is then full, of length $1$, of which $q$ occupies
$x_q$ and the
remainder is occupied by jobs of size at least $s_2$, whence
$w_1 \ge x_q p_q + (1 - x_q) s_2$. Using
$\sum_{z \ge 3} s_z \le \sum_{z=2}^{K_i - 1} w_z = L'_i - w_1 - w_{K_i}$
from Lemma~\ref{lem:chain}(i),
\[
\text{load} \ \le\ p_q + s_2 + L'_i - w_1 - w_{K_i}
\ \le\ (1 - x_q) p_q + x_q s_2 + L'_i - w_{K_i}
\ <\ (1 - x_q) + \tfrac{x_q}{2} + 1 ,
\]
using $p_q \le 1$, $s_2 \le 1/2$, $L'_i \le 1$ and $w_{K_i} > 0$. So the
load is $< 2 - x_q/2$. Because $q$ was not assigned in Step~2, both of
its shares are below $\beta = 2/3$, so $x_q > 1/3$ and the load is
$< 2 - \tfrac16 = \tfrac{11}{6}$.

\emph{(S) A Step-2 job of size $p$.} Every slot job is small, so
$s_1 \le 1/2$; the Step-2 job contributes at least $\tfrac23 p$ to the
relaxation load of $i$, so $L'_i \le 1 - \tfrac23 p$. Lemma~\ref{lem:chain}(ii)
gives load $< p + \tfrac12 + 1 - \tfrac23 p = \tfrac32 + \tfrac{p}{3}
\le \tfrac{11}{6}$, since $p \le 1$.

In every case the load is strictly below $11/6$. The final sentence
follows: Theorem~\ref{thm:ratio} gives ratios $11/6 - \delta$ for every
small $\delta > 0$, and no run attains $11/6$.
\end{proof}

\noindent
Figure~\ref{fig:supremum} puts the two halves side by side: the family that
approaches $\tfrac{11}{6}$ and the ceiling that keeps it out of reach.

\begin{figure}[htbp]\centering
\includegraphics[width=0.97\textwidth]{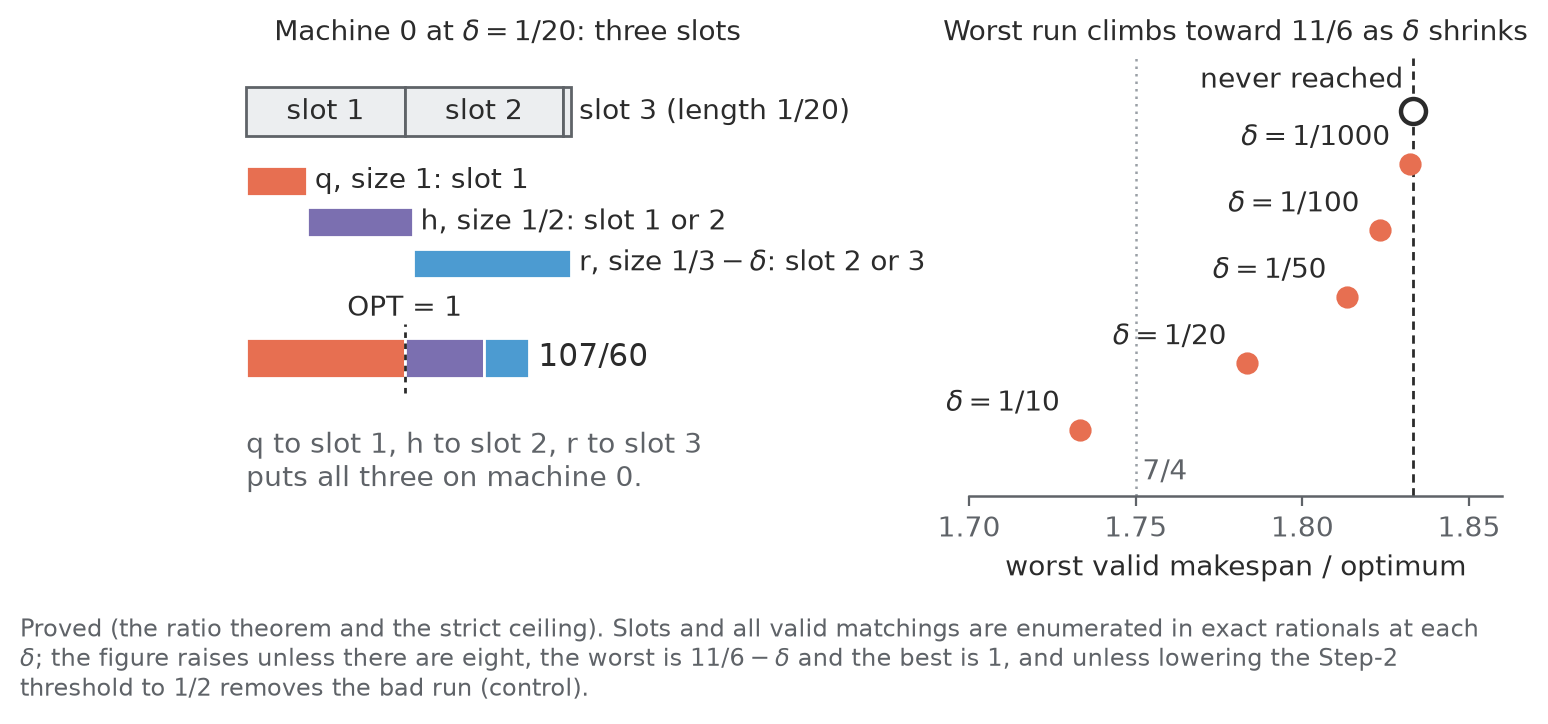}
\caption{Approaching $\tfrac{11}{6}$ without attaining it.
\textbf{Left:} the instance of Theorem~\ref{thm:ratio} at $\delta = \tfrac1{20}$.
Machine~$0$ carries fractional mass $2 + \delta$, so it opens three slots, the
last only $\delta$ long; the jobs $q$, $h$, $r$ meet slot~$1$, slots~$1$--$2$ and
slots~$2$--$3$, and the valid matching that sends them to slots $1$, $2$, $3$
stacks all three on machine~$0$, for load $\tfrac{11}{6} - \delta$ against
$\OPT = 1$. \textbf{Right:} the worst of the eight valid matchings at five values
of $\delta$, each enumerated in exact arithmetic. The family supplies every ratio
below $\tfrac{11}{6}$, and Proposition~\ref{prop:ceiling} keeps every machine
strictly below $\tfrac{11}{6}T$, so $\tfrac{11}{6}$ is a supremum and not a
maximum.}
\label{fig:supremum}
\end{figure}

\noindent
Kept parametric instead of relaxed to $11/6$, the same case split gives a bound that is sometimes much sharper and names the instance data that control it.

\begin{proposition}[an interpolating ceiling]\label{prop:interp}
Fix an instance, a feasible $x$ at $T$, and Step~2's output. Write
$b_{\max} = \max_i b_i$ over the Step-2 jobs ($0$ if there are none),
$p_{\max}$ for the largest job entering the slot structure, and
\[
  \sigma \;=\; \max\{\, s_2(i) \;:\; i \text{ carries no Step-2 job and }
  K_i \ge 2 \,\},
\]
with $p_{\max} := 0$ if no job enters the slot structure and $\sigma := 0$ if
the set defining it is empty --- two conditions, not one, since the slot
structure can be non-empty while no machine qualifies for $\sigma$. Then \emph{every} valid
matching has makespan
\[
  <\ \max\Bigl\{\ \tfrac32 T + \tfrac{b_{\max}}{3},\
                  T + \tfrac{2p_{\max} + \sigma}{3}\ \Bigr\} ,
\]
written so that both sides are makespans: $b_{\max}$, $p_{\max}$ and $\sigma$
are sizes, so a factor of $T$ outside the braces would be dimensionally wrong
except under the normalization $T = 1$ this note uses everywhere else. Both
terms are at most $\tfrac{11}{6}T$. In particular, if
$2p_{\max} + \sigma \le \tfrac94$ and $b_{\max} \le \tfrac34$, then every
valid matching is within $\tfrac74 T$.
\end{proposition}

\begin{proof}
Partition the machines exactly as in Proposition~\ref{prop:ceiling} and
keep each case's bound unrelaxed. A machine with no slot job carries load at most $T$. (Machines with
$K_i = 1$ need no separate treatment --- the three cases already cover them ---
and they are not bounded by $T$: a Step-2 job of size $1$ on a machine whose
only slot takes a job of size $\tfrac12$ reaches $\tfrac32$.) Case (S) gives
$\tfrac32 + b_i/3 \le \tfrac32 + b_{\max}/3$. Case (N0) has no big job
matched to slot~$1$, so the parametrised bound below is not even defined
there; its own bound is $< \tfrac32 \le \tfrac32 + b_{\max}/3$, and it is
absorbed by the first branch. Case (N1) gives
$1 + p_q - x_q(p_q - s_2)$, which decreases in $x_q$ because
$p_q > 1/2 \ge s_2$, so $x_q > 1/3$ bounds it by
$1 + (2p_q + s_2)/3 \le 1 + (2p_{\max} + \sigma)/3$ (for $K_i = 1$ there is no
$s_2$, and the load is $p_q \le T$, absorbed by the first branch). Taking the
maximum
over the cases gives the display, and $p_{\max} \le 1$, $\sigma \le 1/2$,
$b_{\max} \le 1$ give the two ceilings.
\end{proof}

\noindent
\emph{This proposition depends on the half-open convention exactly as
Proposition~\ref{prop:ceiling} does.} Two steps of the proof use it: the
monotonicity step reads ``decreases in $x_q$ because $p_q > 1/2 \ge s_2$'',
and the final ceiling uses $\sigma \le 1/2$. Both are the same fact ---
$s_z \le \tfrac12$ for $z \ge 2$ --- which is exactly what the closed reading
breaks, since under it a big job whose fraction ends on the slot boundary
becomes adjacent to slot~2. Under the closed reading $\sigma$ can reach
$\tfrac{99}{100}$ on the Preliminaries' own counterexample, and then the
monotonicity step, the claim that both terms are at most $\tfrac{11}{6}T$, and
the $\tfrac74$ corollary all fail together.

\begin{remark}
Case (N1) with $x_q > 1/3$ recovers Wang and Sitters' Case~2 bound, and
the family of Theorem~\ref{thm:ratio} realizes it: there
$x_q = 1/3 + \delta$, the bound reads $11/6 - \delta/2$, and the
achieved load is $11/6 - \delta$. The three sizes $1$, $\tfrac12$,
$\tfrac13$ are the three terms of the Wang--Sitters case analysis --- the
slot-$1$ job may be as large as $T$, the slot-$2$ job is bounded by the
previous slot's mass, and so on --- realized simultaneously on one
machine, which their proof leaves open.
\end{remark}

\begin{measurement}\label{meas:corroborate}
Two independent corroborations, both from the deposited audit.
(a) A $13$-machine, $23$-job instance found by plateau walk
(\texttt{p29amb\_ratio\_\allowbreak audit\_\allowbreak walk\_\allowbreak 6000\_\allowbreak s20260802.json}) has $\OPT = 2$,
relaxation threshold $2$, and makespan $3.6571$ under the faithful algorithm with
\emph{GLOP's own} relaxation solution: ratio $1.82855$, with no adversarial
choice made anywhere. The \emph{ratio} reproduces on re-execution; the instance
does not. The walk seeds itself from the artifacts already present in its
directory, so a clean re-run reaches an equally extremal but different instance,
and which one depends on what sits in its directory when it runs. Run alone
from \texttt{osf\_p29sup.zip} it changes one job size ($0.3457$ where the
deposited artifact has $0.3292$); run by that bundle's driver, after another
walk has written beside it, it reaches a $22$-job instance. Every rerun we made
returned the ratio $1.82855$. The claim therefore rests on the deposited instance and its independent
certificate rather than on re-running the search: the instance is recorded
exactly, and certified in the companion note \cite{Paper2}. This is
evidence that the phenomenon of Theorem~\ref{thm:ratio} is not an
artifact of hand-picking $x$. That instance is no longer only evidence:
its solver solution is now certified a vertex and its numbers are exact,
and it is the certified vertex witness of \cite{Paper2}.
(b) Over $4{,}000$ random instances with random-objective relaxation solutions,
taking the scheme's own Step-3 matching on each instance,
the largest ratio of makespan to the linear-programming threshold was
$1.6918$, and $1.7143$ on a second seed. These runs record ratios only; they do
not evaluate the per-machine bounds of Proposition~\ref{prop:ceiling}, which are
checked in the companion note \cite{Paper2} and nowhere else. Both
figures are read from
\texttt{p29amb\_ratio\_audit\_random\_4000\_s20260801.json} and its
\texttt{s20260802} companion, and the first reproduces bit-for-bit on
re-execution.
\end{measurement}

\begin{remark}[a number in the shipped artifacts that does not reproduce]
The artifact \texttt{p29amb\_ratio\_\allowbreak audit\_\allowbreak walk\_\allowbreak 6000\_\allowbreak s20260801.json}
records a ratio $1.9993$ against the relaxation's target for the faithful
algorithm. That exceeds $11/6$ and so contradicts \cite{WS16}; on
re-execution the same instance gives $1.4895$. We treat the recorded
value as a numerical artifact of the bisected target, and no claim here
rests on it. We report audit maxima below only where we have
re-executed them.
\end{remark}

\section{A weaker family, tight only against the relaxation}\label{sec:lprel}

The next family is tight against the relaxation's target $T$ but not against the
optimum. We keep it because its relaxation solution is essentially forced, which
makes it a cleaner illustration of the slot mechanism.
Figure~\ref{fig:tight} draws both panels of what follows: the family
under an arbitrary valid matching, and the same instance under the
best-matching variant.

\begin{figure}[htbp]\centering
\includegraphics[width=0.95\textwidth]{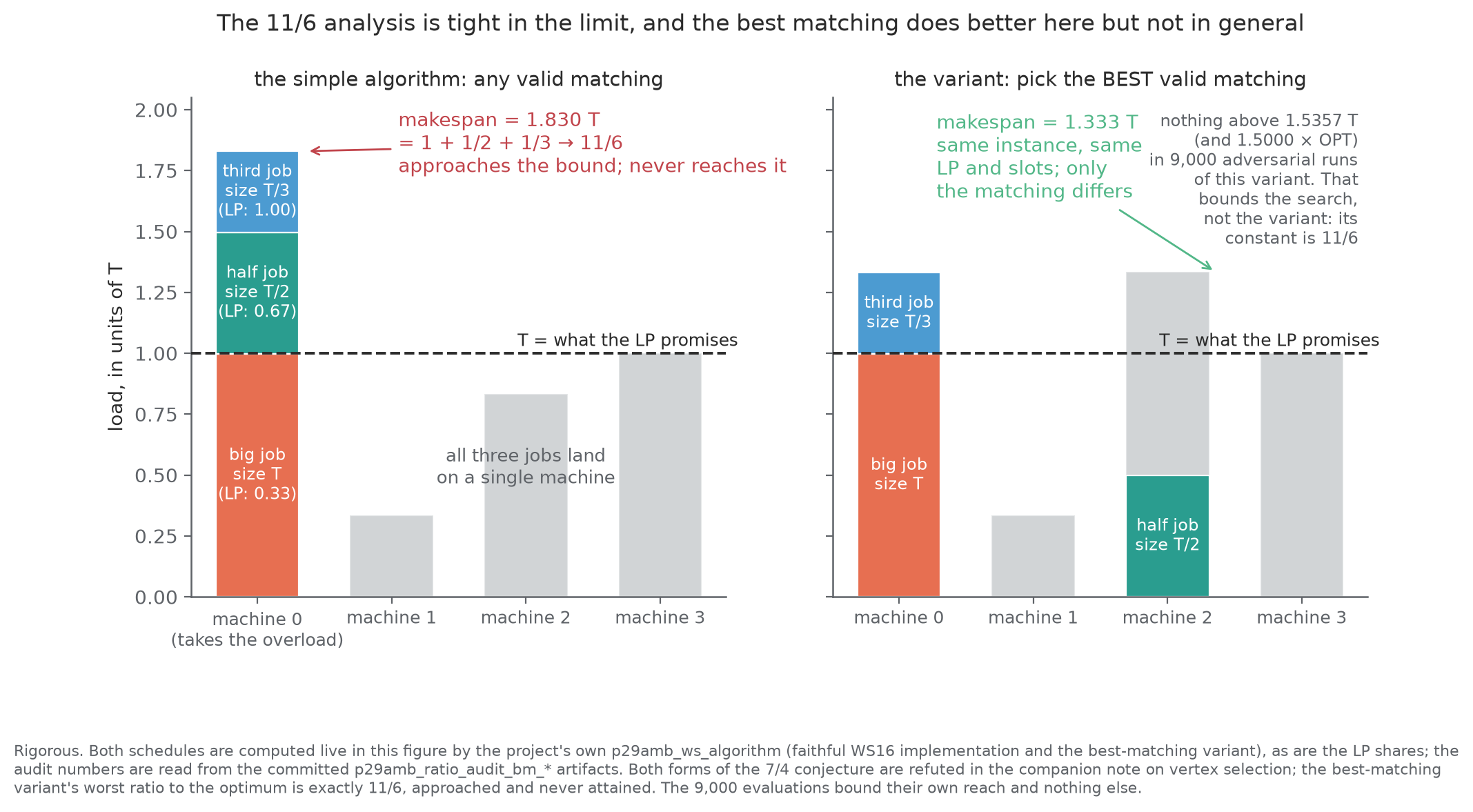}
\caption{The six-job family of Theorem~\ref{thm:tight} (left): the three
slots of one machine absorb the jobs of sizes $1$, $\tfrac12$ and $\tfrac13$
at target $T = 1 + 2\delta$, the three terms of the Wang--Sitters analysis,
realized at once. The figure is drawn in units of $T$ and labels these sizes
$T$, $T/2$ and $T/3$: their ratios to $T$, which approach $1$, $\tfrac12$ and
$\tfrac13$ as $\delta \downarrow 0$. Right: the same instance
under the best-matching variant. The figure is drawn against the relaxation
target $T$; on this family the optimum is $\approx \tfrac43 T$, so the left panel
shows relaxation-relative tightness, not an approximation ratio of $11/6$. For
the approximation ratio see Theorem~\ref{thm:ratio}, whose three-job
family has the same slot picture with the forced singleton jobs that saturate
the partner machines removed.}
\label{fig:tight}
\end{figure}

\begin{theorem}\label{thm:tight}
For every $\delta \in (0, 1/15]$ there is a six-job instance on four machines,
feasible for the relaxation at $T = 1 + 2\delta$ and at no smaller
target, on which some valid matching yields makespan
$(1 + \tfrac12 + \tfrac13)T - O(\delta) = \tfrac{11}{6}T - O(\delta)$.
Moreover the feasible region is contained in a box of width $O(\delta)$: any feasible
$x$ has
\[
x_{q0} \in [\tfrac13 + \delta,\ \tfrac13 + \tfrac32\delta], \qquad
x_{h0} \in [\tfrac23 + \delta,\ \tfrac23 + 2\delta], \qquad
x_{r0} = 1 .
\]
The instance does \emph{not} witness an approximation ratio of $11/6$:
its optimum is $\OPT = \tfrac43 + \tfrac52\delta > T$, so the ratio
realized against the optimum is
$\tfrac{11}{6} / (\tfrac43 + \tfrac52\delta) \to \tfrac{11}{8} = 1.375$
as $\delta \downarrow 0$.
\end{theorem}

\begin{proof}
Take machines $\{0,1,2,3\}$ and jobs: $q$ of size $1$ on $\{0,1\}$; $h$
of size $1/2$ on $\{0,2\}$; $r$ of size $1/3$ on $\{0,3\}$; and three
forced singletons, of sizes $T - (2/3 - \delta)$ on machine $1$,
$T - \tfrac12(1/3 - \delta)$ on machine $2$, and $T$ on machine $3$.

The singleton on machine $3$ has size $T$, so $x_{r3} = 0$ and
$x_{r0} = 1$. The singleton on machine $1$ leaves capacity
$2/3 - \delta$ there, so $x_{q1} \le 2/3 - \delta$ and
$x_{q0} \ge 1/3 + \delta$. The singleton on machine $2$ leaves capacity
$\tfrac12(1/3 - \delta)$, and $h$ has size $1/2$, so
$x_{h2} \le 1/3 - \delta$ and $x_{h0} \ge 2/3 + \delta$. Machine $0$'s
load constraint reads $x_{q0} + \tfrac12 x_{h0} + \tfrac13 \le T$, i.e.
$x_{q0} + \tfrac12 x_{h0} \le \tfrac23 + 2\delta$; combined with the two
lower bounds this forces $x_{q0} \le \tfrac13 + \tfrac32\delta$ and
$x_{h0} \le \tfrac23 + 2\delta$, giving the stated box. That box is derived from
\emph{necessary} conditions, so it does not by itself say the relaxation is
feasible; one point exhibits that. Take
$x_{q0} = \tfrac13 + \delta$, $x_{h0} = \tfrac23 + \delta$ and $x_{r0} = 1$,
so $x_{q1} = \tfrac23 - \delta$ and $x_{h2} = \tfrac13 - \delta$, each exactly
at the capacity its singleton leaves. Machine $0$'s load is then
\[
\bigl(\tfrac13 + \delta\bigr) + \tfrac12\bigl(\tfrac23 + \delta\bigr) + \tfrac13
\ =\ 1 + \tfrac32\delta \ \le\ 1 + 2\delta \ =\ T ,
\]
and its big-job budget is $x_{q0} = \tfrac13 + \delta \le 1$, since $q$ is the
only big job. So the relaxation is feasible at $T$; and it is infeasible
below $T$ because of the size-$T$ singleton.

Machine $0$'s fractional mass is $x_{q0} + x_{h0} + 1 \in
(2,\, 2 + \tfrac72\delta]$, so it carries three slots. Sorted by
nonincreasing size, $q$ occupies $[0, x_{q0})$, $h$ occupies
$[x_{q0}, x_{q0} + x_{h0})$ and $r$ the next unit, so $q$ touches slot
$1$, $h$ touches slots $1$ and $2$, and $r$ touches slots $2$ and $3$.
Sending $q, h, r$ to slots $1, 2, 3$ is a valid matching putting all
three on machine $0$, for load
$1 + \tfrac12 + \tfrac13 = \tfrac{11}{6}$ against $T = 1 + 2\delta$,
i.e.\ $\tfrac{11}{6}T - O(\delta)$.

It remains to compute $\OPT$. The schedule $q \mapsto 0$,
$h \mapsto 2$, $r \mapsto 3$ has loads $1$ on machine $0$,
$\tfrac13 + 3\delta$ on machine $1$,
$T - \tfrac12(\tfrac13 - \delta) + \tfrac12 = \tfrac43 + \tfrac52\delta$
on machine $2$, and $T + \tfrac13 = \tfrac43 + 2\delta$ on machine $3$,
so $\OPT \le \tfrac43 + \tfrac52\delta$. Conversely, every schedule puts
$q$ on machine $0$ or $1$ and $h$ on machine $0$ or $2$. If
$q \mapsto 1$, machine $1$ carries $1 + \tfrac13 + 3\delta =
\tfrac43 + 3\delta > \tfrac43 + \tfrac52\delta$. If $q \mapsto 0$ and
$h \mapsto 0$, machine $0$ carries at least $\tfrac32$, which exceeds
$\tfrac43 + \tfrac52\delta$ for $\delta < 1/15$ and equals it at
$\delta = 1/15$. This is where the theorem's hypothesis is needed: above
$1/15$ this branch becomes the cheaper one and $\OPT$ is no longer
$\tfrac43 + \tfrac52\delta$. If $q \mapsto 0$ and
$h \mapsto 2$, machine $2$ carries exactly
$\tfrac43 + \tfrac52\delta$. So $\OPT = \tfrac43 + \tfrac52\delta$,
which exceeds $T = 1 + 2\delta$, and the realized ratio against the
optimum is $\tfrac{11}{6}/(\tfrac43 + \tfrac52\delta)$.
\end{proof}

\begin{measurement}\label{meas:sixjob}
At $\delta = 10^{-3}$: $T = 1.002$, relaxation threshold $1.002$, algorithm
makespan $1.8333333$, so $\ALG/T = 1.8296740$; but $\OPT = 1.3358333$
and therefore $\ALG/\OPT = 1.3724267$. At $\delta = 10^{-2}$ and
$10^{-4}$ the corresponding ratios to the optimum are $1.3496933$ and
$1.3747422$. The $O(\delta)$-box of Theorem~\ref{thm:tight} is what its proof
establishes; the deposit does not regenerate it separately, and this note
claims no independent numerical check of it.
Under the best-matching oracle the same instances give makespan exactly
$\OPT$, so the ratio to the optimum is exactly $1.0000$; the familiar $4/3$ is
the ratio to $T$, and even that only in the limit --- at $\delta = 10^{-3}$ it
is $\OPT/T = (\tfrac43 + \tfrac52\delta)/(1 + 2\delta) = 8015/6012 = 1.33317$,
tending to $\tfrac43$ from below as $\delta \downarrow 0$. \texttt{p29amb\_meas\_sixjob\_check.py} regenerates these
ratios; it is the only entry point that prints them.
\end{measurement}

\begin{remark}[the family is not $\beta$-independent]\label{rem:beta}
This family does not bind for every admissible
$\beta$. The big job's shares are $x_{q0} \approx
1/3 + \delta$ and $x_{q1} \approx 2/3 - \delta$, so Step~2 assigns it
exactly when $\beta \le 2/3 - \delta$, and then the ratio collapses.
Measured at $\delta = 10^{-3}$, sweeping $\beta$ over
$\{0.50, 0.51, \dots, 0.99\}$: for $\beta \le 0.66$ the big job is
assigned in Step~2 and $\ALG/T = 1.3336660$ ($\ALG/\OPT = 1.0003743$);
for $\beta \ge 0.67$ it is not, and $\ALG/T = 1.8296740$
($\ALG/\OPT = 1.3724267$). Wang and Sitters allow $1/2 < \beta < 1$, so
the family binds on $(2/3 - \delta,\, 1)$ only. The family of
Theorem~\ref{thm:ratio} has the same limitation, for the same reason.
Neither family shows that no choice of $\beta$ helps.
Proposition~\ref{prop:beta}, at the end of Section~\ref{sec:sy}, settles the
question against the relaxation's value, for every $\beta \in (\tfrac12, 1)$; Theorem~\ref{thm:optexact} closes it against the optimum there as well, and Theorem~\ref{thm:belowfamily} closes it against the optimum at every $\beta \le \tfrac12$. It
is stated there because its proof needs that section's reading of
\cite{SY21}'s two instances.
\end{remark}

\section{Relation to Schwartz and Yeheskel}\label{sec:sy}

Schwartz and Yeheskel \cite{SY21} study graph balancing with
\emph{orientation costs}, and place \cite{WS16} as the $(11/6, 3/2)$
point of a bicriteria frontier. Their Lemma~12 shows the analysis of
their Algorithm~2 with threshold function $f_\alpha$ is tight on the
threshold scale, at
$\max\{3/2 + \alpha/2,\ 5/2 - \alpha\}$, which is $11/6$ at
$\alpha = 2/3$, and it exhibits two instances (their Figure~4).

The constructions remain tight against the relaxation target after removing orientation costs. We
rebuilt both Figure-4 instances in
plain graph balancing with all orientation costs deleted, ran Step~2 with
$\beta = \alpha$, and enumerated every valid slot matching.

\begin{measurement}\label{meas:sy}
Everything below is \texttt{p29amb\_sy\_figure4.py}, which rebuilds both
instances from the text of their Lemma~12, deletes the orientation costs
and changes nothing else. Their proof states the relaxation's solution is
unique at $T = 1$---each vertex carries fractional load exactly $1$ and the
graph is acyclic---so the slot structure is determined and is built in exact
rational arithmetic rather than read off a floating-point solve; uniqueness is taken from \cite{SY21}'s own argument rather than inferred
numerically, and solving independently in GLOP returns the same point, which is
used here as corroboration. Makespans are exact rationals and the valid
matchings are enumerated in full.
\emph{(a)} At $\alpha = \beta = 2/3$ and $\epsilon = 10^{-3}$, with costs
deleted:
Figure~4(a) ($5$ machines, $13$ jobs after splitting the self-loops as
\cite{SY21} do; $33$ valid matchings) has worst valid matching
$1.8328333$ against $T = 1$, matching their claimed
$1.5 + 0.5\alpha - 0.5\epsilon$ exactly; Figure~4(b) ($3$ machines, $8$
jobs, $6$ valid matchings) has worst valid matching $1.8323333$ against
$T = 1$, matching their claimed $2.5 - \alpha - \epsilon$
exactly.\footnote{The proof of their Lemma~12 concludes ``the makespan is
$2.5\alpha - \epsilon$''; we read this as a typo for $2.5 - \alpha -
\epsilon$, which is what the lemma's own statement and their Figure~4
caption give, and what the instance in fact realizes; $2.5\alpha -
\epsilon$ would be $1.666$ at $\alpha = 2/3$, against the measured
$1.832$.}
\emph{(b)} Sweeping $\alpha \in \{0.50, 0.51, \dots, 0.99\}$ with
$\beta = \alpha$, both reproduce their claimed value at every $\alpha$; in
particular Figure~4(b) binds for every $\alpha \in [1/2, 1)$, covering the
whole range on which our own families fail (Remark~\ref{rem:beta}). One
caveat: Figure~4(a) places the big edge at fraction \emph{exactly} $\alpha$
on $v_1$, so Wang and Sitters' non-strict test $x_{ij} \ge \beta$ assigns it
in Step~2 and the instance collapses to $1.3338$; with the threshold applied
strictly, as in \cite{SY21}'s Local Step, it binds as stated. Figure~4(b)
binds under the non-strict test as well, but only on the decreasing branch,
so above $\alpha = 2/3$ it gives less than $11/6$.
\emph{(c)} The edit of Proposition~\ref{prop:beta}---moving a load $\sigma$
from $q_u$ into $q_{v_1}$---restores Figure~4(a) under the non-strict test.
At $\alpha \in \{0.67, 0.70, 0.75, 0.80, 0.90, 0.99\}$ and
$\sigma \in \{10^{-4}, 10^{-3}, 10^{-2}\}$, all $18$ shifted instances have
a unique relaxation solution, leave the big edge unassigned by Step~2, and
have worst valid matching exactly $\tfrac32 + \tfrac\alpha2 -
\tfrac\epsilon2 - \sigma$; that value exceeds $11/6$ in every case except
$\alpha = 0.67$ with $\sigma = 10^{-2}$, where $\sigma$ is larger than the
available headroom $\tfrac\alpha2 - \tfrac13 = 0.0017$. In all $18$ the Shmoys--Tardos
rounding in nonincreasing size order, run in the deposited script's own
traversal and with no adversarial tie-break, returns that worst value itself.

The phrasing avoids calling this \emph{the deterministic algorithm}, because it is not one. Shmoys--Tardos fixes the
job order and the slot cutting; it does not fix which augmenting path the
matching step takes. On the certified vertex witness of \cite{Paper2} the
same rounding gives $36571/20000$ scanning machines one way and $26571/20000$
the other --- the worst valid matching against the best, a swing of exactly
$\tfrac12$ --- and on the six-job family of Theorem~\ref{thm:tight} it gives
$4009/3000$ in one traversal order and $\tfrac{11}{6}$ in another. So naming a
standard algorithm for Step~3 does not remove the freedom this note is about;
it moves it into the matching routine's traversal, where it is easier to miss
because the algorithm has a name. A claim about ``the'' Shmoys--Tardos output
needs its traversal stated, and every such claim here states it.
Deposited as \texttt{p29amb\_sy\_\allowbreak figure4\_\allowbreak
perturbed.json} and \texttt{p29amb\_sy\_\allowbreak figure4\_\allowbreak
sweep.json}.
\end{measurement}

So the relaxation-relative phenomenon of Theorem~\ref{thm:tight} is not new: it
is \cite{SY21}'s Lemma~12 with the costs removed, and their statement is
the stronger one, holding for every $\alpha \in [1/2,1)$ where ours holds
only above $2/3 - \delta$. Theorem~\ref{thm:ratio}
remains new: neither Figure-4 instance is tight against the
\emph{optimum}. Measured at $\alpha = 2/3$, $\epsilon = 10^{-3}$, their
optima are $1.3333333$ and $1.3323333$, giving worst-case ratios
$\ALG/\OPT$ of $1.3746250$ and $1.3752815$: the same $1.37$ regime as
our six-job family, and far from $11/6$.

The distinction affects how the following sentence should be read. Opening their Section~3.2, shortly after Lemma~12, they
write: ``It is important to note that Lemma~12 implies that using
Algorithm~2 with \textbf{LP} and the threshold function $f_\alpha$ cannot
achieve an approximation better than $11/6$ with respect to the makespan.''

Two readings are available, and this section is about the difference
between them. Schwartz and Yeheskel define an $(\alpha,\beta)$-approximation
relative to a \emph{target} $T$---``given a target makespan $T$, it outputs
an orientation with makespan at most $\alpha T$''---and they locate $T$ by
binary search on the relaxation's feasibility. Read in their own terms, the
sentence says that with \textbf{LP} as the relaxation no threshold function
$f_\alpha$ certifies better than $11/6$; Lemma~12 gives that immediately,
the sentence is true, and the next clause---``To this end we strengthen
\textbf{LP}''---is exactly the response it calls for. The strengthening concerns the same object this note measures, and its definition is worth recording. \cite{SY21} keep \textbf{LP} and add, for every vertex $u$ and every
$S \subseteq \delta(u)$ with $\sum_{e \in S} p_e > T$ and $|S| \le k$,
\[
  \sum_{e \in S} x_{e,u} \;\le\; |S| - 1
  \tag{Set}
\]
---in their words, ``if a collection of edges $S$ touching $u$ has total weight of
more than $T$ then not all edges in $S$ can be chosen''---and they use $k = 3$, so
no separation oracle is needed. They add that ``these constraints cannot be
inferred from the original relaxation of [3]'', which is \cite{EKS14}. At $|S| = 2$
(Set) says that two edges whose weights sum above $T$ cannot both be oriented into
$u$: a per-vertex constraint on jobs of \emph{any} size, where \cite{EKS14}'s
(Star $v$) constrains only the big ones. We think this is the
natural reading, and on it there is nothing to repair. Read instead as a
statement about $\ALG/\OPT$, the sentence is still true, but Lemma~12 does
not prove it: Lemma~12 is tight against $T$, and on both Figure-4
instances the worst valid matching reaches $1.83$ times $T$ and only $1.37$
times $\OPT$. Schwartz and Yeheskel show elsewhere that they state
$\OPT$-relative bounds explicitly when they mean them: their Theorem~6
reads ``every orientation has a makespan of at least $11/6 - \epsilon$''.

\begin{remark}[what a per-vertex strengthening can buy, and what it cannot]\label{rem:strengthen}
The natural response to everything below is to strengthen the relaxation until the
freedom costs less, and (Set) is one instance of exactly that. It is useful to say
in one place how far that can go, because the answer is known and it bounds this
note's subject.

\cite{VW14} characterise these cuts in general and bound their strength in one
step. One family of them, for any $\alpha \in \mathbb{N}^{+}$: if the jobs of a
set $S$ have total processing time on $i$ above $\alpha T$, then at least $\alpha$ of
them run elsewhere, so $\sum_{j \in S} x_{ij} \le |S| - \alpha$. (Set) is the case
$\alpha = 1$. And their Proposition~1 characterises \emph{every} cut that speaks
about one machine at a time --- any valid $\sum_j \alpha_j x_{ij} \le \delta_{\alpha,i}$
with $\delta_{\alpha,i} = \max\{\sum_{j \in S} \alpha_j : S$ feasible on $i\}$ --- and
proves all of them are implied by the configuration linear program. \cite{EKS14}'s
(Star $v$) is that family with $\alpha$ the indicator of the big jobs; so is any
other per-vertex cut one might write down here.

Neither paper states it this way: \emph{(Star $v$) is
already a clique inequality}. The big edges at $v$ pairwise conflict --- two of
them exceed $T$ together --- so they form a clique in the graph on $\delta(v)$
that joins two edges when their weights sum above $T$, and (Star $v$) is that
clique's inequality. The clique is not in general maximal: a medium edge
conflicting with every big one extends it. So the relaxation this note works with
already contains one instance of a family it does not name, at one particular
non-maximal clique.

Two consequences, and neither is about this note's own results. First, no strengthening of this shape can make the relaxation stronger than the configuration program, whose integrality gap for graph balancing lies between $3/2$ --- witnessed, by the smallest instance within two-weight graph balancing (sizes in $\{1,2\}$), in \cite{Shavit26} --- and $1.749$ \cite[Thm.~1]{JR19}, the upper end being a gap bound rather than an approximation ratio, as the Introduction records. Second, and in the other direction, the configuration
program already implies every such cut, since no single configuration contains two
jobs that cannot share a machine. \textbf{Any per-vertex strengthening of LP therefore lies between
\cite{EKS14}'s relaxation and the configuration program, not necessarily
strictly}, and what this note measures --- what the rounding may do with the
freedom \textbf{LP} leaves --- is a fact about the lower end of that interval.
It is open whether a strengthening in that interval improves the rounding's
constant, and this note does not answer it.
\end{remark}

The statement that survives is narrower, and it is what this section uses: no
instance in \cite{SY21} witnesses $\ALG/\OPT$ near $11/6$ for plain graph
balancing: their Theorem~6 is an integrality gap for the larger model
$\mathrm{GB}^u(0.5)$, their Theorem~3 is bicriteria and relative to
$\mathrm{LP}_k$, and their Figure-4 instances are at $1.37$ against the
optimum. Theorem~\ref{thm:ratio} constructs that witness.

One identification belongs here as a qualification rather than a claim.
Their Algorithm~2 with $f_\alpha$ is
close to the \cite{WS16} family with $\beta = \alpha$ but not a member of it,
since at a share exactly equal to $\alpha$ it can leave a big job unseized
where every execution of that family seizes it: their Local Step tests $x > f(p)$ strictly where
\cite{WS16} test $x_{ij} \ge \beta$---the difference Measurement~\ref{meas:sy}
encounters on \cite{SY21}'s Figure~4(a)---and their Global Step takes a
\emph{minimum-cost} matching where \cite{WS16} take any valid one. With the
costs deleted the second difference disappears; the first does not.

Their $11/6$ integrality-gap lower bound (their Theorem~6) is for
$\mathrm{GB}^u(0.5)$ --- unrelated weights and hyperedges --- a strictly
larger model than the one here, in which the gap of the relaxation used
in this note is $7/4$ \cite[\S4.1]{EKS14}.

\begin{proposition}[the threshold is already optimal]\label{prop:beta}
For $\beta \in (\tfrac12, 1)$ let $\mathcal{E}_{\beta,T_{\mathrm{LP}}(I)}(I)$ be the execution space
of the scheme that differs from \cite{WS16} only in using threshold $\beta$ in
Step~2, run at the least feasible target $T_{\mathrm{LP}}(I)$, and let
\[
  R_T(\beta) \;=\; \sup_{I}\ \sup_{(x,M)\,\in\,\mathcal{E}_{\beta,T_{\mathrm{LP}}(I)}(I)}
  \ \frac{\ALG(I, x, M)}{T_{\mathrm{LP}}(I)}
\]
be its guarantee against the relaxation's own value. Then
\[
  R_T(\beta) \ \ge\ g(\beta) \ :=\
  \max\{\tfrac32 + \tfrac\beta2,\ \tfrac52 - \beta\}
  \qquad\text{for every } \beta \in (\tfrac12, 1),
\]
and $R_T(\tfrac23) = \tfrac{11}{6}$. Since $g$ is minimized only where its two
branches cross, $g(\beta) > \tfrac{11}{6}$ for every $\beta \ne \tfrac23$, so no
choice of $\beta$ has a guarantee below $\tfrac{11}{6}$ and $\beta = \tfrac23$
is the unique minimizer of $R_T$ on $(\tfrac12, 1)$. Both statements are
quantified over the whole interval; nothing here is sampled.
\end{proposition}

\begin{proof}
The lower bound is carried by two families, each covering an interval of
thresholds rather than a point. Both are \cite{SY21}'s Figure-4 instances with
the orientation costs deleted, the second one perturbed as below. Neither the
instances nor their range of thresholds is new (\cite{SY21}'s Lemma~12 holds for
every $\alpha \in [1/2, 1)$); what is new is the load-shifting edit that makes
them bind under Wang and Sitters' non-strict test. The
reason an interval argument is available at all is that every number in either
family --- job sizes, fractional shares, self-loop weights --- is an
\emph{affine} function of $\beta$, so each decision the scheme makes is a
comparison between affine functions: which jobs Step~2 seizes, how many slots
$\lceil\,\cdot\,\rceil$ gives a machine, and which slots a job's fraction
touches. Fix the sign of finitely many such comparisons and the slot structure
is combinatorially the same at every $\beta$ in the interval, with each
machine's load a single affine function of $\beta$. We therefore never take a
maximum over matchings: we name one valid matching and one machine, and bound
that machine's load, which is all a lower bound on $R_T$ needs.

\emph{Two preliminaries, stated because they are what make the families
witnesses rather than examples.} First, in each family the exhibited $x$ loads
every machine to exactly $1$. Summing $\sum_j x_{ij} p_j = 1$ over the $m$
machines and using $\sum_i x_{ij} = 1$ gives $\sum_j p_j = m$, so at any
$T < 1$ the load constraints already fail and $T_{\mathrm{LP}} = 1$: the ratios
below are against the least feasible target, as $R_T$ requires. Second, $R_T$
takes a supremum over the feasible solutions Step~1 may return, so exhibiting
one feasible $x$ suffices; we do not need \cite{SY21}'s uniqueness argument
here. It survives as a robustness remark --- it is why an implementation that
solves the relaxation its own way meets the same witness, which
Measurement~\ref{meas:sy} checks against GLOP --- not as a step.

\emph{The decreasing branch.} Fix $\beta \in (\tfrac12, 1)$ and
$\varepsilon \in (0, 1 - \beta)$. Take \cite{SY21}'s Figure~4(b) with the costs
deleted: vertices $u, v_1, v_2$; edges $(u,v_1)$ of size $1$ and $(u,v_2)$ of
size $\tfrac12$; self loops $q_u = 1 - \beta - \varepsilon$,
$q_{v_1} = \beta + \tfrac\varepsilon2$,
$q_{v_2} = \tfrac12 + \tfrac\varepsilon2$, each split into two halves as
\cite{SY21} split theirs. Put
$x_{(u,v_1),u} = \beta + \tfrac\varepsilon2$,
$x_{(u,v_1),v_1} = 1 - \beta - \tfrac\varepsilon2$,
$x_{(u,v_2),u} = \varepsilon$, $x_{(u,v_2),v_2} = 1 - \varepsilon$, and each
loop half wholly on its own vertex. Every share is positive because
$\varepsilon < 1 - \beta$, the three machine loads are
$(\beta + \tfrac\varepsilon2) + \tfrac\varepsilon2 + (1 - \beta - \varepsilon)
= 1$, $(1 - \beta - \tfrac\varepsilon2) + (\beta + \tfrac\varepsilon2) = 1$ and
$\tfrac12(1 - \varepsilon) + (\tfrac12 + \tfrac\varepsilon2) = 1$, and the only
big job is $(u,v_1)$, whose fraction is at most $1$ on each of its two
machines; the largest loop half is
$\tfrac12(\beta + \tfrac\varepsilon2) < \tfrac12$. So $x$ is feasible at
$T = 1$.

Step~2 seizes $(u,v_1)$ onto $u$, since
$x_{(u,v_1),u} = \beta + \tfrac\varepsilon2 \ge \beta$, and the choice is not
ambiguous: its other share is below $\beta$ because $\beta > \tfrac12$. That
leaves $u$ carrying a base load of $1$. The survivors with positive share on
$u$ are the half-size edge, of size $\tfrac12$ and share $\varepsilon$, and
$u$'s two loop halves, of size $\tfrac12(1 - \beta - \varepsilon) < \tfrac12$
and share $1$ each. Their total fractional mass is $2 + \varepsilon$, so $u$
gets $\lceil 2 + \varepsilon\rceil = 3$ slots; nonincreasing size puts the
half-size edge first, on $[0, \varepsilon)$, and the loop halves on
$[\varepsilon, 1+\varepsilon)$ and $[1+\varepsilon, 2+\varepsilon)$. Under the
half-open convention of Section~\ref{sec:prelim} the first touches slot~$1$ and
the other two touch $\{1,2\}$ and $\{2,3\}$, so sending them to slots $1$, $2$
and $3$ respectively is valid; the remaining jobs fit because $v_1$'s two loop
halves have fractional mass exactly $2$ and hence a slot each, and $v_2$, whose
mass $3 - \varepsilon$ gives it three slots, needs only two of them once the
half-size edge is matched on $u$. Machine $u$ then carries
\[
  1 \;+\; \tfrac12 \;+\; (1 - \beta - \varepsilon)
  \;=\; \tfrac52 - \beta - \varepsilon .
\]
None of the comparisons above involved $\beta$ except through
$\beta > \tfrac12$ and $\varepsilon < 1 - \beta$, so this holds simultaneously
at every threshold in $(\tfrac12, 1)$. Letting $\varepsilon \downarrow 0$ gives
$R_T(\beta) \ge \tfrac52 - \beta$ for every such $\beta$.

\emph{The increasing branch.} Fix $\beta \in (\tfrac12, 1)$ and
$\varepsilon, \sigma > 0$ with $\varepsilon < 1 - \beta$ and
$\sigma < \min\{2\beta - 1,\ \tfrac\beta2 - \tfrac\varepsilon2\}$. Take
\cite{SY21}'s Figure~4(a) with the costs deleted and one load moved, which is
the edit Measurement~\ref{meas:sy}\,(c) reports: vertices $u, v_1, v_2, u', v'$;
edges $(u,v_1)$ and $(u',v')$ of size $1$ and $(u,v_2)$ of size $\tfrac12$;
self loops
$q_u = \tfrac\beta2 - \tfrac\varepsilon2 - \sigma$,
$q_{v_1} = 1 - \beta + \sigma$,
$q_{v_2} = \tfrac12 + \tfrac\beta2 + \tfrac\varepsilon2$,
$q_{u'} = 1 - \beta - \varepsilon$, $q_{v'} = \beta + \varepsilon$, halved as
before; and
$x_{(u,v_1),u} = 1 - \beta + \sigma$, $x_{(u,v_1),v_1} = \beta - \sigma$,
$x_{(u,v_2),u} = \beta + \varepsilon$,
$x_{(u,v_2),v_2} = 1 - \beta - \varepsilon$,
$x_{(u',v'),u'} = \beta + \varepsilon$,
$x_{(u',v'),v'} = 1 - \beta - \varepsilon$. At $\sigma = 0$ this is theirs.
All five loads are $1$; no job has size above $1$; the only big jobs are the two
size-$1$ edges, since the largest loop half is
$\tfrac14 + \tfrac\beta4 + \tfrac\varepsilon4 < \tfrac12$ when
$\beta + \varepsilon < 1$; and each of those two puts a fraction at most $1$ on
each of its machines. So $x$ is feasible at $T = 1$.

The shift is what the non-strict test costs them, and it is why $\sigma$ is
here. At $\sigma = 0$ the edge $(u,v_1)$ sits at fraction \emph{exactly}
$\beta$ on $v_1$, so \cite{WS16}'s test $x_{ij} \ge \beta$ seizes it and the
instance collapses --- Measurement~\ref{meas:sy}\,(b) reports the collapsed
value $1.3338$. With $\sigma > 0$ both of its shares are strictly below
$\beta$: the share at $v_1$ is $\beta - \sigma$, and the share at $u$ is
$1 - \beta + \sigma < \beta$ exactly because $\sigma < 2\beta - 1$. So Step~2
leaves that edge free and seizes only $(u',v')$, onto $u'$. The survivors with
positive share on $u$ are then $(u,v_1)$ of size $1$ and share
$1 - \beta + \sigma$, the half-size edge of share $\beta + \varepsilon$, and
$u$'s two loop halves of size
$\tfrac\beta4 - \tfrac\varepsilon4 - \tfrac\sigma2 < \tfrac12$ and share $1$
each. Their mass is $3 + \varepsilon + \sigma$, so $u$ gets $4$ slots, and
nonincreasing size lays them on $[0, 1-\beta+\sigma)$,
$[1-\beta+\sigma,\ 1+\varepsilon+\sigma)$,
$[1+\varepsilon+\sigma,\ 2+\varepsilon+\sigma)$ and
$[2+\varepsilon+\sigma,\ 3+\varepsilon+\sigma)$, touching slots $\{1\}$,
$\{1,2\}$, $\{2,3\}$ and $\{3,4\}$; sending them to slots $1, 2, 3, 4$ is
valid. The other machines' survivors fit for the same reason as before: $v_1$
and $v_2$ each get three slots and need two, and $u'$ and $v'$ each get two and
need two. Machine $u$ then carries
\[
  1 \;+\; \tfrac12 \;+\; \bigl(\tfrac\beta2 - \tfrac\varepsilon2 - \sigma\bigr)
  \;=\; \tfrac32 + \tfrac\beta2 - \tfrac\varepsilon2 - \sigma ,
\]
again at every $\beta$ in the interval at once. Letting $\varepsilon$ and
$\sigma$ tend to $0$ gives $R_T(\beta) \ge \tfrac32 + \tfrac\beta2$.

Neither family is sensitive to which form of the threshold test is used: in
both, the seized job's share exceeds $\beta$ strictly and the free job's shares
fall short of it strictly, so \cite{SY21}'s strict Local Step
$x > f(p)$ and \cite{WS16}'s $x_{ij} \ge \beta$ agree on every instance
above. The two branches together give $R_T(\beta) \ge g(\beta)$ on
$(\tfrac12, 1)$.

\emph{The value at $2/3$, and the minimizer.} At $\beta = \tfrac23$ the
guarantee is exactly $\tfrac{11}{6}$: Proposition~\ref{prop:ceiling} keeps
every machine strictly below $\tfrac{11}{6}T$ on every instance and every
execution, and Theorem~\ref{thm:ratio}'s family reaches $\tfrac{11}{6} - \delta$
for every $\delta$, on instances where $T_{\mathrm{LP}} = \OPT = 1$. The first
branch of $g$ increases and the second decreases, so $g$ is minimized where
they cross: $\tfrac32 + \tfrac\beta2 = \tfrac52 - \beta$ at $\beta = \tfrac23$,
with $g(\tfrac23) = \tfrac{11}{6}$, and the crossing is transversal, so
$g(\beta) > \tfrac{11}{6}$ strictly for every other $\beta$. Hence for every
$\beta \in (\tfrac12,1) \setminus \{\tfrac23\}$,
\[
  R_T(\beta) \ \ge\ g(\beta) \ >\ \tfrac{11}{6} \ =\ R_T(\tfrac23),
\]
which is both halves of the statement: no $\beta$ does better than
$\tfrac{11}{6}$, and $\tfrac23$ is the only threshold attaining it. The
minimizer is unique over the interval and not merely over a sample, because the
strict inequality above is proved at every $\beta$ and not checked at some.
\end{proof}

The proposition bounds $R_T$ from below, and one scope remark belongs with it:
the suprema $\tfrac52 - \beta$ and $\tfrac32 + \tfrac\beta2$ are approached as
$\varepsilon$ and $\sigma$ vanish, not realized by any single instance; the
families give $g(\beta)$ in the limit, in the same way Theorem~\ref{thm:ratio}
gives $\tfrac{11}{6}$ in the limit.

On its own, the proposition does not bound $R_T$ from \emph{above}
away from $\beta = \tfrac23$, and until now the note did not, because
Proposition~\ref{prop:ceiling} is a statement about $\beta = \tfrac23$. The missing upper half needs only bookkeeping: its proof never uses the value
$\tfrac23$ except at the last line of each case, so carrying $\beta$ through
instead of substituting it gives the ceiling at every threshold in the interval
--- and the ceiling that comes out is $g(\beta)$ itself.

\noindent
\textbf{Attribution.} The value of the next proposition is \emph{not} new. \cite{WS16}'s own analysis carries $\beta$
through both of its cases and ends at
``$\max\{2.5 - \beta,\ 1.5 + \beta/2\} = 11/6$'' for $\beta = 2/3$: the
$\beta$-dependent bound is theirs, and only its specialisation is what their
theorem states. The proposition below adds the case analysis made
explicit --- including the case their two-case argument folds together --- the
\emph{strictness} of every inequality, which is what
Corollary~\ref{cor:betaexact} needs to call $g(\beta)$ an unattained supremum,
and the matching families of Proposition~\ref{prop:beta}, which turn a bound into an exact value
(their instances and the interval form are \cite{SY21}'s; the load-shifting edit that makes them bind under the non-strict test is new). We state it as a proposition rather than
quoting it because the strict form and the third case are used repeatedly below.

\begin{proposition}[the ceiling at every threshold above one half, after \cite{WS16}]\label{prop:betaceil}
Let $\beta \in (\tfrac12, 1)$, let the relaxation be feasible at $T$, let $x$ be
any feasible solution, and let any valid matching be taken in Step~3 of the
scheme run at threshold $\beta$. Then every machine's load is strictly below
$g(\beta)\,T$, with $g(\beta) = \max\{\tfrac32 + \tfrac\beta2,\ \tfrac52 - \beta\}$
as in Proposition~\ref{prop:beta}.
\end{proposition}

\begin{proof}
Normalize $T = 1$ and fix a machine $i$. Above $\tfrac12$ no machine holds two
Step-2 jobs, since two shares of at least $\beta$ sum to more than the big-job
budget of $1$; so the three cases of Proposition~\ref{prop:ceiling} are
exhaustive here for the same reason they are there, and we take them in turn
with $\beta$ carried rather than set to $\tfrac23$.

\emph{(N0) No Step-2 job, and the slot-$1$ occupant is small or absent.}
Nothing in that case mentions the threshold: the load is below
$\tfrac32$, which is below $\tfrac32 + \tfrac\beta2$.

\emph{(N1) No Step-2 job, and a big job $q$ is matched to slot $1$.} The case's
own computation gives load $< 2 - x_q/2$, and every step of it --- all big
fraction in slot~$1$, hence $s_z \le \tfrac12$ for $z \ge 2$; slot~$1$ full;
Lemma~\ref{lem:chain}(i) on the rest --- is a fact about the sort order and the
budget, with no $\beta$ in it. The threshold controls the last step
alone: $q$ was not seized, so both its shares are below $\beta$, and since they
sum to $1$ we get $x_q > 1 - \beta$ and load $< 2 - \tfrac{1-\beta}{2}
= \tfrac32 + \tfrac\beta2$.

\emph{(S) A Step-2 job of size $p$.} The machine carries no unassigned big
fraction, by \cite{WS16}'s observation in the form the general threshold needs:
an unassigned big job $q$ on $i$ has $x_{iq} > 1 - \beta$, the Step-2 job
already holds at least $\beta$, and $\beta + (1 - \beta) = 1$ exhausts the
budget. So $s_1 \le \tfrac12$, and the Step-2 job contributes at least
$\beta p$ to the relaxation load, leaving $L'_i \le 1 - \beta p$.
Lemma~\ref{lem:chain}(ii) then gives
load $< p + \tfrac12 + 1 - \beta p = \tfrac32 + (1 - \beta) p \le \tfrac52 - \beta$,
using $p \le 1$.

The largest of the three bounds is $g(\beta)$, and each is strict.
\end{proof}

\noindent
In words: above one half the worst machine is either one that seized a big job,
which costs $\tfrac52 - \beta$, or one that had to round a big job it could not
seize, which costs $\tfrac32 + \tfrac\beta2$; a high threshold makes the first
cheap and the second dear, and a low one does the reverse. The two costs are
equal at $\beta = \tfrac23$, and that crossing --- not the number $\tfrac{11}{6}$
--- is why the published threshold is the right one. Since the same two
expressions are what Proposition~\ref{prop:beta}'s families realize from below,
the guarantee is pinned.

\begin{corollary}[the guarantee above one half, exactly]\label{cor:betaexact}
$R_T(\beta) = g(\beta)$ for every $\beta \in (\tfrac12, 1)$, and the value is a
supremum that no execution attains. In particular
$R_T(\tfrac23) = \tfrac{11}{6}$, which is Proposition~\ref{prop:ceiling}
together with Theorem~\ref{thm:ratio}.
\end{corollary}

\begin{proof}
Proposition~\ref{prop:beta} gives $R_T(\beta) \ge g(\beta)$ and
Proposition~\ref{prop:betaceil} gives $R_T(\beta) \le g(\beta)$, the latter with
strict inequality on every execution, so the supremum is $g(\beta)$ and is not
attained.
\end{proof}

Measurement~\ref{meas:sy} is now a control on the algebra above rather than the
evidence for it. Its fifty-point sweep and its eighteen shifted instances agree
with the two closed forms exactly, as an algebraic control; but
the proposition no longer rests on them, and its quantifier over
$(\tfrac12, 1)$ is discharged by the interval argument, not by the grid.

Our own Theorem~\ref{thm:ratio} gives less on the $\beta$-separation axis and
more on the optimum-relative one. Its family has big-job shares
$\tfrac13 + \delta$ and $\tfrac23 - \delta$, so Step~2 does not fire for any
$\beta > \tfrac23 - \delta$; the scheme is then Step~3 alone and returns
$\tfrac{11}{6} - \delta$ against $\OPT = T = 1$. That is $\tfrac{11}{6}$ and no
more, so it does not separate $\beta > 2/3$ from $\beta = 2/3$; but it holds
against the \emph{optimum}, where the edited Figure~4(a) reaches only about
$1.38$. That left the threshold question open on the optimum-relative scale,
and the next theorem closes it.

\begin{lemma}[the nested collector family]\label{lem:nested}
Let $\alpha \in (0, 1]$ and let $m \ge 1$ be an integer with
$m\alpha \le 1$ and $m\alpha > \tfrac12$. For every
$\varepsilon \in (0, \tfrac12)$ there is an instance with
$T_{\mathrm{LP}} = \OPT = T$, and on it a legal execution of the scheme at any
threshold $\beta$ with $0 < \beta \le \alpha$, whose makespan is
\[
m(1-\varepsilon) + \tfrac12 + b ,
\qquad
b \ =\ 1 - m\alpha(1-\varepsilon) - \tfrac{\varepsilon}{2} ,
\]
which collects to
\[
\tfrac32 + (1-\alpha)m \ -\ \varepsilon\bigl(m(1-\alpha) + \tfrac12\bigr) ,
\]
strictly below $\tfrac32 + (1-\alpha)m$ for every $\varepsilon > 0$ and
increasing to it as $\varepsilon \downarrow 0$. The execution is legal under the
strict reading $x_{ij} > \beta$ as well whenever $\beta < \alpha$.
\end{lemma}

\begin{proof}
Normalize $T = 1$ and write $p = 1 - \varepsilon$ for the big-job size and
$\mu = \varepsilon$ for the rider's sliver of fraction. Since
$b = (1 - m\alpha) + \varepsilon\bigl(m\alpha - \tfrac12\bigr)$ and
$m\alpha \in (\tfrac12, 1]$, both terms are nonnegative and the second is
positive, so $b > 0$; and
$b - \tfrac12 = -\bigl(m\alpha - \tfrac12\bigr)(1 - \varepsilon) < 0$, so
$b < \tfrac12$. The filler is therefore not big and sorts strictly after the
rider, which is what makes the slot structure below unambiguous.

The instance has $m + 4$ machines --- a \emph{collector} $0$, partners
$1, \dots, m$, a home $h$ for the rider, a home $f$ for the filler, and a pin
machine --- and $m + 3$ jobs:
\begin{itemize}
\item $m$ big jobs of size $p$, the $i$th allowed on $\{0, i\}$, with shares
      $\alpha$ on the collector and $1 - \alpha$ on its partner. They are big
      because $\varepsilon < \tfrac12$;
\item the \emph{rider}, of size exactly $\tfrac12$ --- not big --- allowed on
      $\{0, h\}$, with shares $\mu$ on the collector and $1 - \mu$ on $h$;
\item one \emph{filler}, of size $b$, allowed on $\{0, f\}$, with share $1$ on
      the collector and share $0$ on $f$;
\item a \emph{pin} of size exactly $1$, alone on its machine.
\end{itemize}

\emph{Feasibility, threshold and optimum.} The collector's load is
$m\alpha p + \tfrac{\mu}{2} + b = 1$ exactly, by the definition of $b$, and its
big-job budget is $m\alpha \le 1$. Partner $i$ carries load $(1-\alpha)p$ and
big budget $1 - \alpha$; $h$ carries $(1-\mu)/2$; $f$ carries nothing; the pin
machine carries exactly $1$. So the displayed solution is feasible at $T = 1$,
and the pin --- which no smaller target can place --- gives
$T_{\mathrm{LP}} = 1$. Sending each big job to its own partner, the rider to
$h$, the filler to $f$ and the pin to its machine schedules the instance at
makespan $1$, so $\OPT = 1$ as well.

\emph{The execution.} Every big job's collector share is $\alpha \ge \beta$, so
a legal Step-2 assignment puts all $m$ of them on the collector; if
$\beta < \alpha$ the same assignment is legal under the strict reading. Their
fractions are deleted from $x$. The collector keeps the rider's
$\mu$ and the filler's $1$, a fractional mass of $1 + \mu$, so
$\lceil 1 + \mu\rceil = 2$ slots open there. Nonincreasing size puts the rider
first, because $\tfrac12 > b$ \emph{strictly}, so the rider occupies
$[0, \mu)$ and the filler $[\mu, 1 + \mu)$: the rider touches slot~$1$ only,
and the filler touches both. The matching that sends the rider to slot~$1$ and
the filler to slot~$2$ is therefore valid; it is also the only complete one,
since the rider has no other slot to go to. The collector finishes at
\[
mp + \tfrac12 + b
\ =\ \tfrac32 + (1-\alpha)m \ -\ \varepsilon\bigl(m(1-\alpha) + \tfrac12\bigr) ,
\]
which is the claim. Every other machine carries at most $1$.
Figure~\ref{fig:collector} draws the collector at one choice of the parameters.
\end{proof}

\begin{figure}[htbp]\centering
\includegraphics[width=0.97\textwidth]{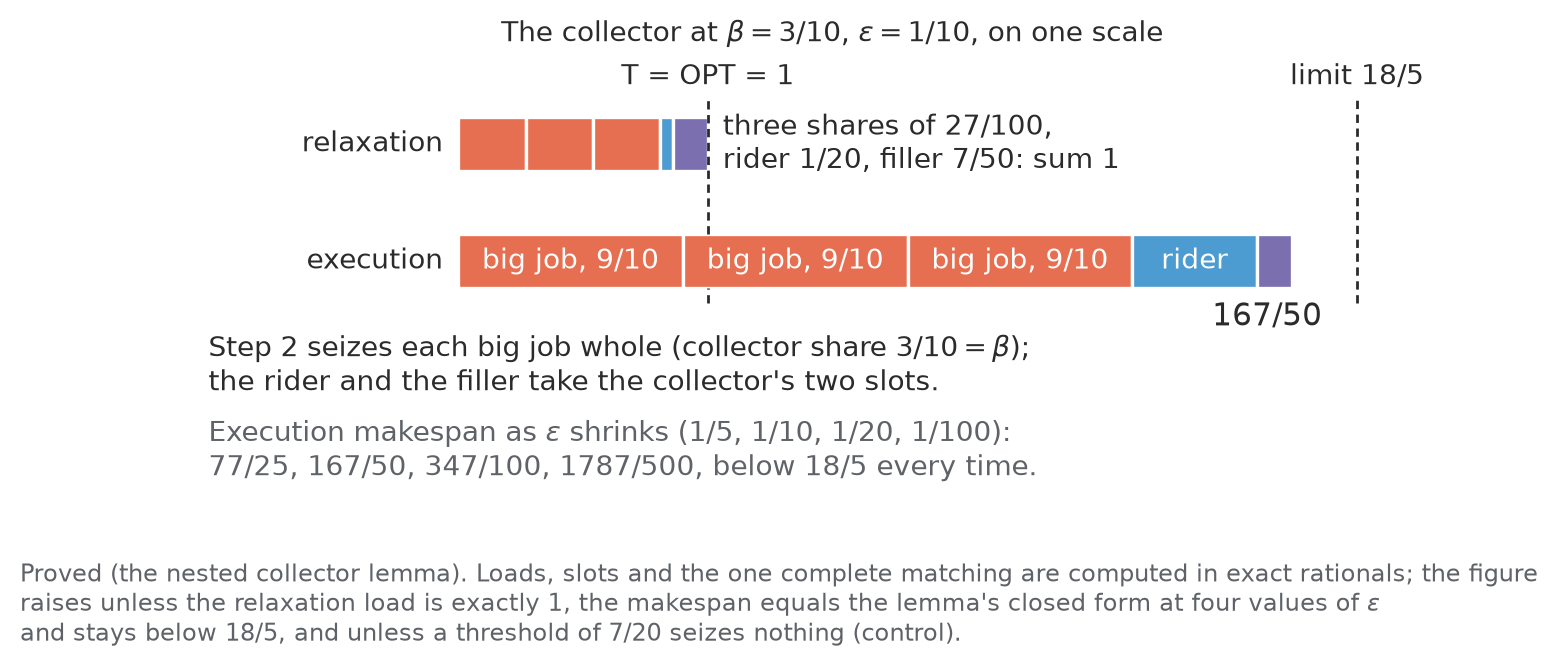}
\caption{The collector of Lemma~\ref{lem:nested} at $\alpha = \beta = \tfrac3{10}$,
$m = 3$ and $\varepsilon = \tfrac1{10}$, both bars on one scale. The relaxation
charges the collector $\tfrac{27}{100}$ for each big job, its share times its
size, so the three big jobs, the rider's sliver and the filler fit at load
exactly $T = 1$. Step~2 may seize each big job whole, since its share on the
collector reaches $\beta$, and the schedule then charges $\tfrac9{10}$ for each;
the rider and the filler take the collector's two slots. The collector finishes
at $\tfrac{167}{50}$, against $\OPT = 1$, and as $\varepsilon$ shrinks the
makespan rises toward $\tfrac32 + (1-\beta)\lfloor 1/\beta \rfloor =
\tfrac{18}{5}$ without reaching it.}
\label{fig:collector}
\end{figure}

\noindent
One construction with two parameters covers three regimes. Taking $\alpha = \beta$ and
$m = \lfloor 1/\beta\rfloor$ gives the family that closes the interval below one
half (Theorem~\ref{thm:belowexact}); taking $m = 1$ gives the decreasing branch
above it, immediately below; and taking $\alpha = \beta + 2\delta$ with
$m = \lceil 1/\beta\rceil - 1$ gives the strict-rule family of
Proposition~\ref{prop:strictbelow}. The details the earlier draft of this note left
unsaid --- how many fillers there are, what each one's fraction is, and which
one is adjacent to which slot --- do not arise, because there is one filler
and its fraction is the whole of what the seized jobs and the rider leave.

\begin{theorem}[above one half, against the optimum]\label{thm:optexact}
For every $\beta \in (\tfrac12, 1)$ the guarantee against the \emph{optimum} is
also exactly $g(\beta) = \max\{\tfrac32 + \tfrac\beta2,\ \tfrac52 - \beta\}$, a
supremum that no execution attains.
\end{theorem}

\begin{proof}
Proposition~\ref{prop:betaceil} gives the upper half, with strictness: the
scheme is run at $T = T_{\mathrm{LP}}(I)$, so
\[
  \ALG \ <\ g(\beta)\, T_{\mathrm{LP}} \ \le\ g(\beta)\, \OPT ,
\]
the second step by $T_{\mathrm{LP}} \le \OPT$. Writing the substitution out
leaves no room to ask why a bound at some feasible target should bound the
ratio to the optimum: the target is not just any feasible one. The missing ingredient was a family binding against the optimum. Proposition~\ref{prop:beta}'s families pin the least feasible
threshold by saturating their partner machines with forced singletons, and
those singletons inflate the optimum: the ratio to $T$ survives, the ratio to
$\OPT$ does not. Replace that wall with a \emph{pin} --- one job of size exactly
$T$, alone on its own machine. It forces $T_{\mathrm{LP}} = T$ just as well,
since at any $T' < T$ the relaxation cannot place it, and it costs the optimum
nothing, because it has a machine to itself. Each branch of $g$ then has a
family with $\OPT = T$.

\emph{The decreasing branch, $\tfrac12 < \beta \le \tfrac23$.} Apply
Lemma~\ref{lem:nested} with $\alpha = \beta$ and $m = 1$: the hypotheses hold
because $m\alpha = \beta \in (\tfrac12, 1]$. The collector seizes one big job at
share exactly $\beta$, the rider takes the first slot and the filler the second,
and the makespan tends to $\tfrac32 + (1-\beta) = \tfrac52 - \beta$, which is
$g(\beta)$ on this branch. The lemma is stated for every admissible
$(\alpha, m)$ rather than only below one half, so this branch applies it within
its stated domain and does not extend it.

\emph{The increasing branch, $\tfrac23 \le \beta < 1$.} Use the same three-rung chain as
Theorem~\ref{thm:ratio}, together with a private pin of size $T$ --- the pin
replaces the forced singletons of \emph{Proposition~\ref{prop:beta}}'s families,
not anything in Theorem~\ref{thm:ratio}, whose three jobs are $q$, $h$ and $r$
and which needs no wall. Normalize $T = 1$. Fix $\delta$ with
$0 < \delta < \min\{2\beta - 1,\ \beta/2\}$, which is a nonempty range because
$\beta > \tfrac12$; the first bound is what makes $q$'s two shares fall
\emph{strictly} below $\beta$, so Step~2 misses it, and the second keeps $r$'s
size positive. Four machines: the
collector $0$, a partner $P$, a home $H$, and the pin's own machine. The jobs
are these.
\begin{itemize}
\item $q$, of size $1$, allowed on $\{0, P\}$, with shares
$x_{q0} = 1 - \beta + \delta$ and $x_{qP} = \beta - \delta$ --- \emph{both}
strictly below $\beta$, so Step~2 misses it and it enters the slot structure.
\item $h$, of size $\tfrac12$, allowed on $\{0, H\}$, with shares
$x_{h0} = \beta$ and $x_{hH} = 1 - \beta$. Its size does not exceed $T/2$, so it
is not big and Step~2 does not consider it whatever its shares.
\item $r$, of size $\tfrac\beta2 - \delta$, allowed on $\{0\}$ alone.
\item the pin, of size $1$, allowed on its own machine alone.
\end{itemize}
The shares on the collector sum to
$(1 - \beta + \delta) + \beta + 1 = 2 + \delta$, so
$\lceil 2 + \delta\rceil = 3$ slots are opened there --- it is this
\emph{fractional mass} that sets the slot count, while the collector's weighted
fractional load is $1 \cdot (1 - \beta + \delta) + \tfrac12 \cdot \beta +
(\tfrac\beta2 - \delta) = 1$ exactly, as feasibility at $T = 1$ requires. The
matching that takes the three jobs in nonincreasing order finishes at
$1 + \tfrac12 + \tfrac\beta2 - \delta \to \tfrac32 + \tfrac\beta2$.

Both families have $\OPT = T$, by an exhibited schedule rather than by each
job fitting somewhere on its own. On the increasing branch send $q \mapsto P$,
$h \mapsto H$, $r \mapsto 0$ and the pin to its own machine: the loads are
$1$, $\tfrac12$, $\tfrac\beta2 - \delta$ and $1$, so the makespan is $1$. On the
decreasing branch the schedule is Lemma~\ref{lem:nested}'s, which sends each big
job to its own partner, the rider to $h$, the filler to $f$ and the pin to its
machine.
So the supremum against the optimum is at least $g(\beta)$ on each branch, and
Proposition~\ref{prop:betaceil} caps it at $g(\beta)$ strictly.
\end{proof}

\begin{measurement}[the optimum-relative families, checked
exactly]\label{meas:optexact}
\texttt{p29sup\_beta\_below\_half.py} builds both families at
$\beta \in \{\tfrac{11}{20}, \tfrac35, \tfrac58, \tfrac23, \tfrac7{10},
\tfrac34, \tfrac45, \tfrac9{10}\}$ in exact rational arithmetic and checks, at
each: that the relaxation solution is feasible at $T$; that a job of size
exactly $T$ pins the threshold; that $\OPT = T$ \emph{by exhaustive enumeration}
rather than by an exhibited schedule; and the makespan over every legal Step-2
assignment and every valid matching. Each run also asserts that the makespan
falls \emph{strictly below} $g(\beta)$ --- a construction that reached it would
be refuting Proposition~\ref{prop:betaceil} rather than confirming the family,
and the check fails rather than reports agreement --- and that it falls short by
less than the perturbation. The shortfalls run from $0.00008$ to $0.001$.

Above one half the two scales therefore agree: the guarantee against the
relaxation's value and against the optimum are both $g(\beta)$. Below one half
they agreed already, since the families of Section~\ref{sec:below} have
$T_{\mathrm{LP}} = \OPT$. So the scheme's worst case is now known against the
optimum at \emph{every} threshold in $(0,1)$, and $\tfrac23$ minimizes it on
both scales.
\end{measurement}

\section{Below one half, the guarantee is unbounded as \texorpdfstring{$\beta \to 0$}{beta -> 0}}\label{sec:below}

Proposition~\ref{prop:beta} and Corollary~\ref{cor:betaexact} settle the
interval $(\tfrac12, 1)$. They do not reach the question of whether a threshold
$\beta \le \tfrac12$ improves the guarantee, and one half is not an arbitrary
place to stop: it is where two facts about the relaxation go slack at once.
A machine's big-job budget is $1$, so two shares of at least $\beta$ fit on one
machine exactly when $2\beta \le 1$; and a single big job's two shares sum to
$1$, so both of them reach $\beta$ exactly when $2\beta \le 1$ as well. Above
one half neither can happen and Step~2 is a rule that fires or does not. At one
half and below, both happen, and Step~2 becomes a choice.

That choice has to be added to the object under study. Section~\ref{sec:prelim} writes an execution as a pair $(x, M)$; at
$\beta \le \tfrac12$ it is a triple $(x, A, M)$ with $A$ a legal Step-2
assignment --- legal meaning that each big job goes to a machine where its share
reaches $\beta$. Above one half $A$ is determined by $x$, so this extends the
definition rather than replacing it. It is also not a technicality: the whole of
what follows is what the third coordinate costs.

The answer is that below one half no single constant covers every threshold.
At each fixed $\beta$ the guarantee is a constant, and an exact one; what fails
is uniformity in $\beta$. It grows like $1/\beta$, it is unbounded as the
threshold falls, and it jumps at one half itself. The reason is a count of seized jobs rather than the slot chain.

\begin{lemma}[below one half, Step~2 seizes everything]\label{lem:allseized}
Let $0 < \beta \le \tfrac12$. Then every big job is assigned in Step~2, so every job
entering the slot structure has size at most $\tfrac12$, and $s_z \le \tfrac12$
for every slot on every machine. Cases (N0) and (N1) of
Proposition~\ref{prop:ceiling} are therefore empty.
\end{lemma}

\begin{proof}
A big job's fraction is supported on at most two machines and sums to $1$, so
its larger share is at least $\tfrac12 \ge \beta$ and Step~2's rule fires there.
A job seized in Step~2 is deleted from $x$ on both of its machines, so no big
fraction survives into the slot structure.
\end{proof}

\noindent
So the three-case analysis collapses to one case, and the remaining question is how many seized jobs one machine may receive. The
budget answers it: shares of at least $\beta$ that sum to at most $1$ number at
most $\lfloor 1/\beta \rfloor$, and each can carry a whole job of size up to
$T$.

\begin{theorem}[the ceiling below one half]\label{thm:belowceil}
Let $0 < \beta \le \tfrac12$, write $k = \lfloor 1/\beta \rfloor$, let the
relaxation be feasible at $T$, and let $x$, a legal Step-2 assignment and a
valid matching be arbitrary. Then every machine's load is strictly below
$\bigl(\tfrac32 + (1 - \beta)k\bigr) T$.
\end{theorem}

\begin{proof}
Normalize $T = 1$ and fix a machine $i$. Write $S_i$ for the jobs Step~2 put on
$i$ and $B_i$ for their total size. Every one of them is big and holds at least
$\beta$ of the machine's big-job budget of $1$, so $|S_i| \le k$, and since each
has size at most $1$, $B_i \le k$. They contribute at least $\beta B_i$ to the
machine's relaxation load, so what is left for the slot jobs satisfies
$L'_i \le 1 - \beta B_i$ --- which is the note's own $L'_i \le 1 - \beta b_i$
with the single seized job's size replaced by the total.

If $i$ carries no slot job its load is $B_i \le k$, below the claimed bound
because $\beta k \le 1$. Otherwise $w_{K_i} > 0$ by Lemma~\ref{lem:chain}(i),
and Lemma~\ref{lem:chain}(ii) --- whose proof adds the seized jobs' sizes to the
matched sizes and so holds verbatim with $B_i$ in place of $b_i$ --- gives
\[
\text{load} \ \le\ B_i + s_1 + L'_i - w_{K_i}
\ <\ B_i + \tfrac12 + 1 - \beta B_i
\ =\ \tfrac32 + (1 - \beta) B_i
\ \le\ \tfrac32 + (1 - \beta)k ,
\]
using $s_1 \le \tfrac12$ from Lemma~\ref{lem:allseized} and $w_{K_i} > 0$ for
the strictness. In words: a machine pays for the jobs Step~2 hands it, for one
slot job of size at most half the target, and for whatever load the seized jobs
left it, and the seized jobs pay for themselves at the rate $\beta$.
\end{proof}

\noindent
A family meets the bound, and meets it against the
\emph{optimum} rather than against the relaxation's value.

\begin{theorem}[the family below one half]\label{thm:belowfamily}
Let $0 < \beta \le \tfrac12$ and $k = \lfloor 1/\beta \rfloor \ge 2$. There are instances with $T_{\mathrm{LP}} = \OPT = T$ for which an execution has makespan $(k + \tfrac12)T$ when $\beta k < 1$, and makespans approaching $(k+\tfrac12)T$ when $\beta k = 1$. Hence the guarantee at threshold $\beta$, against the
optimum and so also against the relaxation's value, is at least
$\lfloor 1/\beta \rfloor + \tfrac12$.
\end{theorem}

\begin{proof}
Normalize $T = 1$ and fix $\eta = 0$ if $\beta k < 1$, and a small
$\eta \in (0, 1)$ otherwise. The instance has $k + 3$ machines --- a
\emph{collector} $0$, partners $1, \dots, k$, a home $k+1$, and a machine $k+2$
--- and $k + 2$ jobs:
\begin{itemize}
\item $k$ big jobs of size $1 - \eta$, the $i$th allowed on $\{0, i\}$, with
      shares $\beta$ on the collector and $1 - \beta$ on its partner;
\item a \emph{rider} of size exactly $\tfrac12$, allowed on $\{0, k+1\}$, with
      share $\mu$ on the collector and $1 - \mu$ on $k+1$, where $\mu$ is any
      number with $0 < \mu \le 1$ and $\tfrac{\mu}{2} \le 1 - \beta k (1 - \eta)$;
\item a \emph{pin} of size exactly $1$, allowed only on machine $k+2$.
\end{itemize}
The rider is not big: big means size strictly greater than $T/2$, and the rider
is exactly $T/2$. Such a $\mu$ exists because $\beta k \le 1$, with room bought
by $\eta$ in the case $\beta k = 1$.

\emph{Feasibility at $T = 1$.} The collector's big-job budget is $\beta k \le 1$
by the definition of $k$, and its load is $\beta k (1-\eta) + \tfrac{\mu}{2} \le 1$
by the choice of $\mu$. Partner $i$ carries load $(1-\beta)(1-\eta)$ and big
budget $1 - \beta$, both at most $1$; machine $k+1$ carries at most $\tfrac12$;
machine $k+2$ carries exactly $1$.

\emph{The threshold and the optimum are both $1$.} The pin has size exactly $1$,
so at any $T' < 1$ the relaxation forces its shares to zero and no solution
exists; with the solution above feasible at $1$, the least feasible threshold is
$1$. An integral schedule of makespan $1$ puts each big job on its own partner,
the rider on $k+1$ and the pin on $k+2$; and no schedule does better, since the
pin has size $1$. So $T_{\mathrm{LP}} = \OPT = 1$ and the ratios below are
against the optimum.

\emph{The execution.} Every big job's share on the collector is exactly $\beta$,
so a legal Step-2 assignment may send all $k$ of them there --- the choice is
real, because the partner share $1 - \beta$ also reaches $\beta$. The pin is
seized onto $k+2$. The collector is left with the rider's fraction $\mu$,
giving it one slot, to which the rider is adjacent; a valid matching sends the
rider there rather than to its home. The collector's load is then
$k(1-\eta) + \tfrac12$, and no other machine exceeds $1$. With $\eta = 0$ this
is $k + \tfrac12$ exactly; otherwise it approaches $k + \tfrac12$ as $\eta$
falls, which is all the statement claims.
\end{proof}

\noindent
The difference between the two cases is worth a sentence, because it is the only place the
arithmetic of $\beta$ shows through. If $1/\beta$ is a whole number, the
collector's budget is already tight at $\beta k = 1$, and with big jobs of full
size its load constraint is tight as well, so there is no room left for the
rider's sliver of fraction; shrinking the big jobs by $\eta$ buys the room and
costs $k\eta$ of makespan. If $\beta k < 1$, the slack is there from the start
and the value is reached exactly.

\begin{corollary}[the whole range]\label{cor:wholerange}
Write $R_T(\beta)$ for the guarantee at threshold $\beta$ --- as in
Proposition~\ref{prop:beta} above one half, and with the supremum taken over
legal Step-2 assignments as well at or below it, where they are a third choice.
Then
\begin{enumerate}
\item[(i)] $\beta = \tfrac23$ is the unique minimizer of $R_T$ over all of
$(0,1)$, not only over $(\tfrac12, 1)$: at or below one half
$R_T(\beta) \ge \lfloor 1/\beta\rfloor + \tfrac12 \ge \tfrac52 > \tfrac{11}{6}$.
\item[(ii)] $R_T(\beta)$ is unbounded as $\beta \to 0$, and below one half it is
exactly $\tfrac32 + (1-\beta)\lfloor 1/\beta\rfloor$
(Theorem~\ref{thm:belowexact}).
\item[(iii)] At a reciprocal threshold the two bounds coincide:
$R_T(1/k) = k + \tfrac12$ exactly, for every integer $k \ge 2$, as a supremum
that is approached and not attained.
\item[(iv)] $R_T$ jumps at one half: $g(\beta) \to 2$ as $\beta \downarrow \tfrac12$,
while $R_T(\tfrac12) = \tfrac52$.
\end{enumerate}
\end{corollary}

\begin{proof}
(i) and (ii) are Theorems~\ref{thm:belowceil}, \ref{thm:belowfamily}
and~\ref{thm:belowexact} together with Corollary~\ref{cor:betaexact}, using
$\lfloor 1/\beta\rfloor \ge 2$ below one half. For (iii), $\beta = 1/k$ makes the
ceiling $\tfrac32 + (1 - \tfrac1k)k = k + \tfrac12$, which the family
approaches; it is the case of (ii) in which the two bounds already coincided. For
(iv), $g$ is continuous with $g(\tfrac12^+) = 2$, and (iii) at $k = 2$ gives
$\tfrac52$ at one half itself; the jump is the count $\lfloor 1/\beta \rfloor$
passing from $1$ to $2$.
\end{proof}

\noindent
The reading of all this is that a lower threshold does not trade a little
guarantee for a simpler rule. It gives up the constant. The scheme is also
discontinuous in its own parameter, at the one place the relaxation's two
budget facts go slack together, so ``$\beta$ slightly below one half'' is not a
small perturbation of ``$\beta$ slightly above'' it.

\begin{theorem}[the interval closes: the guarantee below one half]\label{thm:belowexact}
Let $0 < \beta \le \tfrac12$ and $k = \lfloor 1/\beta\rfloor$. The guarantee at
threshold $\beta$, against the optimum, is exactly
$\tfrac32 + (1-\beta)k$ --- a supremum that is approached and never attained.
\end{theorem}

\begin{proof}
Theorem~\ref{thm:belowceil} is the upper half, with the strictness that makes
the supremum unattained. For the lower half apply Lemma~\ref{lem:nested} with
$\alpha = \beta$ and $m = k$. Its two hypotheses hold: $\beta k \le 1$ by the
definition of $k$, and $\beta k > \tfrac12$ because $\beta(k+1) > 1$ forces
$\beta k > 1 - \beta \ge \tfrac12$ at $\beta \le \tfrac12$. The lemma returns instances with $T_{\mathrm{LP}} = \OPT = T$ for which a legal execution has makespan \[ \tfrac32 + (1-\beta)k \ -\ \varepsilon\bigl(k(1-\beta) + \tfrac12\bigr) , \] strictly below the ceiling for every $\varepsilon > 0$ and tending to it as $\varepsilon \downarrow 0$. The convergence is linear in the perturbation, which
is what separates an approach from a structural shortfall.

That family spends exactly what Theorem~\ref{thm:belowfamily}'s left on the
table. There the collector opened one slot and collected $k + \tfrac12$, while
the load budget $1 - \beta k$ remained unused. A full slot costs its own weighted
mass and returns the largest size adjacent to the next one, so budget spent on a
second slot comes back as matched size --- here as the single filler, whose size
is exactly the budget the seized jobs and the rider leave behind.
\end{proof}

\begin{corollary}[the guarantee at every threshold, against the optimum]\label{cor:optwholerange}
Write $R_{\OPT}(\beta)$ for the supremum of $\ALG/\OPT$ at threshold $\beta$,
taken over instances, over feasible solutions returned by Step~1 at the least
feasible target $T_{\mathrm{LP}}(I)$, over valid
matchings taken by Step~3, \emph{and, at or below one half, over the legal
Step-2 assignments as well}. The last of these is not a formality: above one
half Step~2's action is determined by $x$, since no two shares of at least
$\beta$ fit in the budget, but at or below one half a big job can be eligible at
both of its machines and several can be eligible at one, so the assignment is a
third choice the specification leaves open --- and it is the choice
Theorem~\ref{thm:belowfamily}'s families exploit. Then for every
$\beta \in (0,1)$
\[
R_{\OPT}(\beta) \;=\;
\begin{cases}
\tfrac32 + (1-\beta)\lfloor 1/\beta\rfloor, & 0 < \beta \le \tfrac12,\\[4pt]
\max\{\tfrac32 + \tfrac\beta2,\ \tfrac52 - \beta\}, & \tfrac12 < \beta < 1,
\end{cases}
\]
and the supremum is attained at no threshold. In particular
$\inf_{\beta} R_{\OPT}(\beta) = R_{\OPT}(\tfrac23) = \tfrac{11}{6}$, and
$\tfrac23$ is its unique minimizer.
\end{corollary}

\begin{proof}
The two branches are Theorem~\ref{thm:belowexact} and
Theorem~\ref{thm:optexact}, which between them cover $(0,1)$; non-attainment is
part of each. The value at $\tfrac23$ is $\tfrac{11}{6}$ from the upper branch,
and Corollary~\ref{cor:wholerange}(i) gives uniqueness of the minimizer, its
argument depending only on the ordering of the two branches and not on which
reference point they are measured against.
\end{proof}

\noindent
The preceding corollary summarizes everything this note establishes about the threshold. Corollary~\ref{cor:wholerange} says the same thing against the relaxation's target $T$; the corollary above states it against $\OPT$, which is the denominator used in an approximation ratio. The two agree, which is not
automatic: the families
that bind against $T$ saturate their partner machines with forced singletons,
and those singletons inflate $\OPT$, so the optimum-relative branches needed
constructions of their own rather than a remark.

\begin{corollary}[the below-half curve, interval by interval]\label{cor:reciprocal}
Let $k \ge 2$ be an integer. On the interval
$\tfrac1{k+1} < \beta \le \tfrac1k$ we have $\lfloor 1/\beta\rfloor = k$, so
Theorem~\ref{thm:belowexact} reads
\[
R_{\OPT}(\beta) \;=\; \tfrac32 + k(1 - \beta),
\]
affine in $\beta$ with slope $-k$. The curve below one half is therefore
piecewise affine, steepening by one unit of slope at each reciprocal threshold,
and it jumps there: at $\beta = \tfrac1k$ its value is $k + \tfrac12$, while its
limit from above is $k - \tfrac12 + \tfrac1k$, a jump of exactly
$1 - \tfrac1k = 1 - \beta$. The case $k = 1$ of the same formula is the
decreasing branch $\tfrac52 - \beta$ above one half, and the jump it predicts at
$\beta = \tfrac12$ is $\tfrac12$, which is
Corollary~\ref{cor:wholerange}(iv).
\end{corollary}

\begin{proof}
$\lfloor 1/\beta\rfloor = k$ exactly on $(\tfrac1{k+1}, \tfrac1k]$, and
substituting it into Theorem~\ref{thm:belowexact} gives the displayed form.
The jump is the difference between $\tfrac32 + (1-\beta)k$ at $\beta = 1/k$ and
$\tfrac32 + (1-\beta)(k-1)$ in the limit from above, which is $1 - \beta$.
\end{proof}

\noindent
Nothing here is new: it is Theorem~\ref{thm:belowexact} read one interval at a
time. The interval form shows the jumps, each one a machine's big-job budget admitting one more seized job, which a single floor function hides.

\begin{figure}[htbp]\centering
\includegraphics[width=0.92\textwidth]{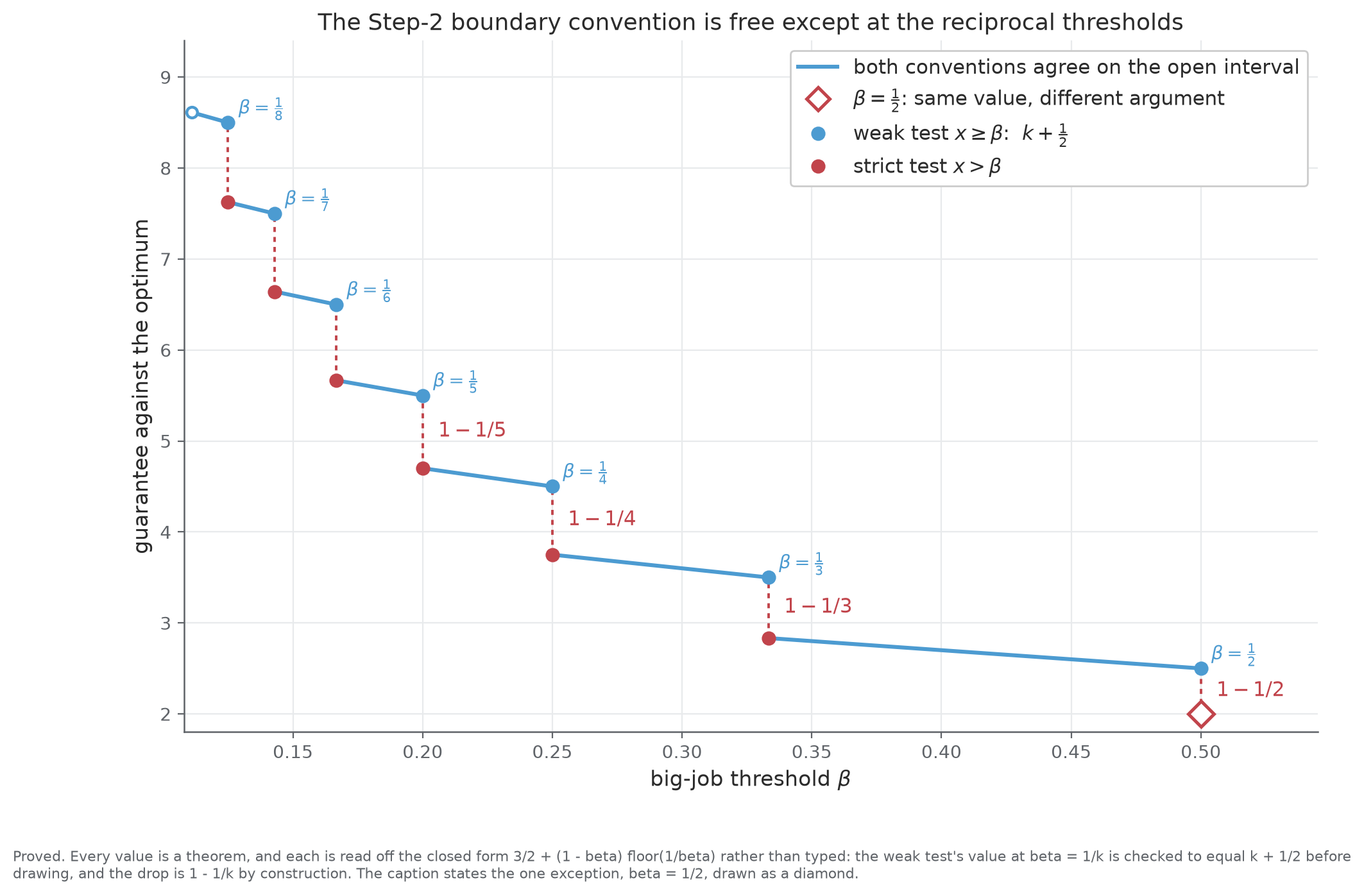}
\caption{The cost of the Step-2 boundary convention, and where it falls. Step~2 seizes a
big job when its share on a machine \emph{reaches} the threshold $\beta$ (the
weak test, $x_{ij} \ge \beta$) or when it \emph{exceeds} it (the strict test,
$x_{ij} > \beta$). The two give the same guarantee on every open interval
between consecutive reciprocal thresholds, drawn as one line because it is one
line: only the count of big jobs a machine's budget admits differs, and off the
reciprocals the two counts are equal. At $\beta = 1/k$ the strict test admits
one job fewer and the guarantee falls by exactly $1 - 1/k$, marked by the dashed
drop. Open circles are limits the interval does not contain; filled circles are
values. \textbf{The one exception is $\beta = \tfrac12$}, drawn as a diamond: a
big job split exactly in half has share exactly $\tfrac12$, which \emph{reaches}
the weak threshold but does not \emph{exceed} the strict one. So
Lemma~\ref{lem:allseized} applies under the weak rule and not under the strict
one, and under the strict rule the value $2$ comes from
Proposition~\ref{prop:betaceil}'s argument above one half rather than from
Theorem~\ref{thm:belowexact}, as Remark~\ref{rem:strictathalf} explains.
Proposition~\ref{prop:strictbelow} supplies the family that meets it, so the
diamond is a value proved from both sides and not a ceiling with nothing
beneath it. Generated by
\texttt{p29sup\_fig\_conventions.py}, which re-checks both claims the picture
makes before drawing it.}
\label{fig:conventions}
\end{figure}

\begin{corollary}[both conventions in one formula]\label{cor:bothconventions}
Write $k_{\ge}(\beta) = \lfloor 1/\beta\rfloor$ and
$k_{>}(\beta) = \lceil 1/\beta\rceil - 1$ for the number of big jobs a machine's
budget admits under the weak test $x_{ij} \ge \beta$ and the strict test
$x_{ij} > \beta$. Then for $0 < \beta \le \tfrac12$ the guarantee under either test is
\[
R_{\ge}(\beta) = \tfrac32 + (1-\beta)\,k_{\ge}(\beta),
\qquad
R_{>}(\beta) = \tfrac32 + (1-\beta)\,k_{>}(\beta) .
\]
That is: the two readings give the same formula, and the only thing the
convention changes is how many big jobs the count admits. Since
$\lceil 1/\beta\rceil - 1 = \lfloor 1/\beta\rfloor$ unless $1/\beta$ is a whole
number, the two agree at every threshold except the reciprocals, where
\[
R_{\ge}(1/k) - R_{>}(1/k) \;=\; 1 - \tfrac1k .
\]
In words: at a reciprocal threshold the strict test admits one big job fewer
than the weak one, and the guarantee falls by what that one job was carrying.
Everywhere else the choice of convention costs nothing at all.
Figure~\ref{fig:conventions} draws both readings; they are one curve except at
the countably many thresholds where they are not, which is the whole content of
this corollary. The endpoint $\beta = \tfrac12$ obeys the same formula but not
by the same argument: under the strict test its ceiling comes from the analysis
above one half rather than from the count, for the reason
Remark~\ref{rem:strictathalf} gives, and its family is
Proposition~\ref{prop:strictbelow}'s.
\end{corollary}

\begin{proof}
The counts are Proposition~\ref{prop:conventions}(ii). The weak reading is
Theorem~\ref{thm:belowexact}; the strict one is
Proposition~\ref{prop:strictbelow}, whose family raises the collector's shares
to $\beta + 2\delta$ and re-spends the load budget that raise costs, rather than
inheriting the weak family unchanged. The two counts differ exactly when
$1/\beta \in \mathbb{Z}$, where $\lceil 1/\beta\rceil - 1 = \lfloor
1/\beta\rfloor - 1$, and the values then differ by $1 - \beta = 1 - 1/k$.
\end{proof}

\begin{remark}[why $\beta = \tfrac12$ is excluded]\label{rem:strictathalf}
Under the strict rule the argument above stops working at exactly one point, and
it is worth saying where rather than carrying the formula through it. The
below-half analysis rests on Lemma~\ref{lem:allseized}, whose proof observes
that a big job's larger share is at least $\tfrac12$ and so reaches the
threshold. Under the weak test that suffices at $\beta = \tfrac12$. Under the
strict test it does not: a big job split exactly $(\tfrac12, \tfrac12)$ has
larger share $\tfrac12$, which is not \emph{greater} than $\tfrac12$, so
Step~2 fires nowhere and the job survives into the slot structure. Cases (N0)
and (N1) of Proposition~\ref{prop:ceiling} are then not empty, and the machinery
of this section does not apply.

The analysis above one half applies instead, and it never used
Lemma~\ref{lem:allseized}: at $\beta = \tfrac12$ under the strict rule no
machine can hold two seized big jobs, since two shares each exceeding $\tfrac12$
sum to more than the budget of $1$, so
Proposition~\ref{prop:betaceil}'s three cases are exhaustive and give
$g(\tfrac12) = \max\{\tfrac32 + \tfrac14,\ \tfrac52 - \tfrac12\} = 2$. The
corollary's formula returns $2$ there as well, so the endpoint needs its own
argument even though the two routes agree on the number. That argument bounds
the endpoint from above only; Proposition~\ref{prop:strictbelow} constructs the
matching family, which completes the equality.
\end{remark}

\begin{proposition}[the strict Step-2 rule, below one half and at it]\label{prop:strictbelow}
Suppose Step~2 fires only when $x_{ij} > \beta$, and let $0 < \beta \le
\tfrac12$. Write $m(\beta) = \lceil 1/\beta\rceil - 1$ for the number of big
jobs a machine's budget then admits. The guarantee against the optimum is
exactly
\[
R_{>}(\beta) \ =\ \tfrac32 + (1 - \beta)\,m(\beta) ,
\]
a supremum approached and never attained. At the endpoint $m(\tfrac12) = 1$ and
$R_{>}(\tfrac12) = 2$.
\end{proposition}

\begin{proof}
\emph{The ceiling below one half.} Let $\beta < \tfrac12$. Every big job is
still seized, since its larger share is at least $\tfrac12 > \beta$, so
Lemma~\ref{lem:allseized} holds verbatim under the strict test. If a machine
receives $r$ seized jobs, each of their shares strictly exceeds $\beta$ and they
sum to at most the budget $1$, so $r\beta < 1$; an integer strictly below
$1/\beta$ is at most $\lceil 1/\beta\rceil - 1$, so $r \le m(\beta)$.
The calculation in Theorem~\ref{thm:belowceil} then applies unchanged with $m(\beta)$ in place of $k$ and puts every load strictly below $\tfrac32 + (1-\beta)m(\beta)$.

\emph{The ceiling at the endpoint.} At $\beta = \tfrac12$
Lemma~\ref{lem:allseized} fails, for the reason
Remark~\ref{rem:strictathalf} gives, and the count argument is unavailable. The
analysis above one half replaces it, with its share inequalities read at the
endpoint: a seized share now exceeds $\tfrac12$, and an unseized big job has
both shares at most $\tfrac12$, hence both exactly $\tfrac12$. So no machine
holds two seized big jobs, since two shares above $\tfrac12$ exceed the budget
of $1$; a machine holding a seized job carries no unassigned big fraction, since
more than $\tfrac12$ plus $\tfrac12$ also exceeds it; and in case (N1)
$x_q = \tfrac12$. The three cases of Proposition~\ref{prop:ceiling} are
therefore exhaustive, and the computations in the proof of
Proposition~\ref{prop:betaceil} bound the load strictly below $\tfrac32$,
$2 - \tfrac{x_q}{2} = \tfrac74$ and $\tfrac32 + \tfrac{p}{2} \le 2$ in the three
cases, so every load is strictly below $g(\tfrac12) = 2$. Since
$\tfrac32 + (1-\tfrac12)\cdot 1 = 2$, the formula holds at the endpoint too, by a
different route.

\emph{The family, at every $\beta \le \tfrac12$ at once.} Fix $\delta > 0$ small
enough that $m(\beta)(\beta + 2\delta) \le 1$, and apply Lemma~\ref{lem:nested}
with $\alpha = \beta + 2\delta$, $m = m(\beta)$ and $\varepsilon = \delta$. The
lemma's hypotheses hold. The budget condition $m\alpha \le 1$ is the choice of
$\delta$. For the other, $m\alpha > m\beta$, and $m\beta \ge \tfrac12$ in both
cases: at a non-reciprocal $\beta$ below one half $m(\beta) =
\lfloor 1/\beta\rfloor$, so $m\beta > 1 - \beta \ge \tfrac12$; at
$\beta = 1/j$ with $j \ge 2$ we have $m(\beta) = j - 1$, so
$m\beta = 1 - 1/j \ge \tfrac12$. Because $\alpha > \beta$ the lemma's execution
is legal under the strict test. Its makespan tends to
$\tfrac32 + (1-\alpha)m(\beta)$ as $\varepsilon \downarrow 0$ and so to
$\tfrac32 + (1-\beta)m(\beta)$ as $\delta \downarrow 0$, and every instance has
$T_{\mathrm{LP}} = \OPT = T$, so the supremum is against the optimum. At
$\beta = \tfrac12$ the family is one big job seized at share
$\tfrac12 + 2\delta$, the rider and the filler, finishing at
$2 - 3\delta + 2\delta^2$.
\end{proof}

\noindent
 Raising the collector's
shares from $\beta$ to $\beta + 2\delta$ is not free: it also changes the load
budget the rider and the filler spend, so the weak-rule family does not simply
``carry through''. Lemma~\ref{lem:nested} takes the share as a parameter
precisely so that the perturbed family is the same construction rather than a
new one, and at the endpoint $\beta = \tfrac12$ it establishes the matching lower
bound, completing the equality there.

\begin{remark}[the strict reading above one half]\label{rem:strictabove}
The claim that $\tfrac23$ is the best threshold under \emph{both} readings needs
one more line, because above one half
Theorem~\ref{thm:optexact}'s decreasing family gives the collector a share of
exactly $\beta$, which the strict test does not seize. Lemma~\ref{lem:nested}
covers that too: take $\alpha = \beta + \rho$ and $m = 1$, admissible because
$m\alpha = \beta + \rho \in (\tfrac12, 1]$ for small $\rho > 0$, and legal under
the strict test because $\alpha > \beta$. Its limit is $\tfrac52 - \alpha$,
which tends to $\tfrac52 - \beta$ as $\rho \downarrow 0$. The increasing
branch needs no repair at all: its family works by having Step~2 \emph{miss} a
big job whose two shares are both strictly below $\beta$, and a strict test only
makes missing easier. So $g(\beta)$ is the guarantee on $(\tfrac12, 1)$ under
both readings, and with Proposition~\ref{prop:strictbelow} below one half,
$\tfrac23$ minimizes the guarantee under both readings across the whole range.
\end{remark}

\begin{measurement}[the repaired constructions, checked exactly]\label{meas:repairs}
\texttt{p29sup\_\allowbreak threshold\_\allowbreak constructions\_\allowbreak verify.py} builds Lemma~\ref{lem:nested}'s family in
exact rational arithmetic at twenty thresholds: nine below one half, four above
it at $m = 1$, and seven in the strict form of
Proposition~\ref{prop:strictbelow}. At each it checks
\begin{itemize}
\item that the displayed solution is feasible at $T$;
\item that $\OPT = T$ by \emph{exhaustive enumeration} of integral assignments,
      not by an exhibited schedule;
\item that a job of size exactly $T$ pins $T_{\mathrm{LP}}$;
\item that the Step-2 assignment is legal under the reading being tested, weak
      or strict;
\item that the collector opens exactly two slots;
\item and that the worst valid matching finishes at the predicted value, and
      strictly below the limit.
\end{itemize}
Four diagnostic conditions accompany them. The one-filler
family must reach the same limit as the superseded spread-filler construction at
every threshold: two constructions disagreeing would put the number itself in
doubt. Every instance built must finish strictly below $\tfrac{11}{6}T$ at
$\beta = 2/3$. The strict reading must be \emph{refused} by the big-job budget
at a reciprocal threshold, which is the one place the two conventions differ.
And
the $m = 1$ instances must land on $\tfrac52 - \beta$, the value
Theorem~\ref{thm:optexact}'s decreasing branch needs. The shortfall falls by a
factor of ten for each decade of $\varepsilon$. At $\beta = \tfrac12$ under the
strict rule the family reaches $1.99997\,T$ at $\delta = 10^{-5}$, which is
$2 - 3\delta + 2\delta^2$ to the digit.

One thing the check does not do: rational arithmetic instantiates rational thresholds only,
so the irrational ones are reached by the proof and not by the check.
\end{measurement}

\begin{remark}[one instance per threshold, in closed form]\label{rem:midpoint}
No limit is needed to see that $\lfloor 1/\beta\rfloor + \tfrac12$ is not the
answer where the two bounds differ. Write $g = 1 - \beta k$ for the slack, which
is positive exactly when $1/\beta$ is not a whole number.
Lemma~\ref{lem:nested} carries one perturbation, so there is one number to set.
Taking
\[
\varepsilon \ =\ \frac{g}{2\bigl(k(1-\beta) + \tfrac12\bigr)}
\]
makes that lemma's execution finish at
\[
k(1-\varepsilon) + \tfrac12 + b \ =\ \bigl(k + \tfrac12\bigr) + g - \varepsilon\bigl(k(1-\beta) + \tfrac12\bigr)
\ =\ \bigl(k + \tfrac12\bigr) + \tfrac{g}{2} ,
\]
the midpoint of the interval, since the ceiling exceeds $k + \tfrac12$ by exactly
$g$. One instance, in closed form, strictly above $\lfloor 1/\beta\rfloor +
\tfrac12$ at every non-reciprocal $\beta \le \tfrac12$ at once. Tying
$\varepsilon$ to $g$ is what makes it uniform: a fixed perturbation loses
$k\varepsilon$ of makespan and would fail just below a reciprocal, where $g$ is
small.
\end{remark}

\begin{proposition}[the two boundary conventions, priced]\label{prop:conventions}
The constructions of this section depend on two boundary conventions: a job is
\emph{big} when its size strictly exceeds $T/2$, and Step~2 fires when a share
is \emph{at least} $\beta$. Their prices differ.
\begin{enumerate}
\item[(i)] \textbf{The size convention is free.} Under the other reading ---
big means size at least $T/2$ --- the rider of size exactly $T/2$ is big, so
Step~2 seizes it on its own home and it never reaches the collector's slot.
Shrinking it to $T/2 - \delta$ restores the family at a cost of $\delta$, so
the supremum is unchanged.
\item[(ii)] \textbf{The share convention costs exactly one seized job, and only
at reciprocal thresholds.} Under the strict reading $x_{ij} > \beta$ the
collector's shares must rise to $\beta + \delta$, so the big-job budget of $1$
admits $\lceil 1/\beta\rceil - 1$ of them rather than $\lfloor 1/\beta\rfloor$.
The two counts agree unless $1/\beta$ is a whole number, where the strict rule
is one smaller. Both bounds move together --- Lemma~\ref{lem:allseized}'s count
gives the ceiling $\tfrac32 + (1-\beta)(\lceil 1/\beta\rceil - 1)$ and the
perturbed family of Proposition~\ref{prop:strictbelow} meets it --- so at $\beta = 1/k$ the
guarantee falls by exactly $1 - \beta$, and nowhere else. The argument relies on Lemma~\ref{lem:allseized}, which the strict rule invalidates at $\beta = \tfrac12$ and only there; Remark~\ref{rem:strictathalf} says why, and Proposition~\ref{prop:strictbelow} settles that endpoint from both sides --- the ceiling by the above-half argument, the family by raising one share above $\tfrac12$.
\end{enumerate}
Neither reading changes the negative result: below one half the guarantee is
unbounded as $\beta$ falls under both, and $\tfrac23$ remains the unique best
threshold under both.
\end{proposition}

\noindent
The two conventions give different guarantee values: the strict share rule
changes the value at the countably many reciprocal thresholds, by $1 - \beta$
each.

\begin{measurement}[the conventions, checked exactly]\label{meas:conventions}
\texttt{p29sup\_beta\_below\_half.py} rebuilds the family under each reading in
exact rational arithmetic. The share convention, at six thresholds: at
$\beta \in \{\tfrac25, \tfrac3{10}\}$ the two rules give the same value, and at
$\beta \in \{\tfrac13, \tfrac14, \tfrac15, \tfrac18\}$ the strict rule drops the
guarantee from $\tfrac72$ to $\tfrac{17}6$, from $\tfrac92$ to $\tfrac{15}4$,
from $\tfrac{11}2$ to $\tfrac{47}{10}$ and from $\tfrac{17}2$ to $\tfrac{61}8$
--- in each case by $1 - \beta$, and in each case the rebuilt family reaches the
new value to within the perturbation. The size convention, at two thresholds:
the two readings differ by $10^{-4}$, the perturbation itself. At
$\beta = \tfrac12$ the strict share rule removes the phenomenon altogether,
leaving the collector one seized job, which is the regime above one half.
\end{measurement}

\begin{measurement}[the family and the ceiling, checked exactly]\label{meas:belowhalf}
At $\beta \in \{\tfrac12, \tfrac25, \tfrac13, \tfrac3{10}, \tfrac14, \tfrac15,
\tfrac18\}$, \texttt{p29sup\_beta\_below\_half.py} verifies in exact rational
arithmetic that the family's relaxation solution is feasible at $T$, that the
threshold is pinned, that $\OPT = T$ by an exhibited schedule, that the Step-2
branch is legal, and that the worst valid matching finishes at
$k(1-\eta) + \tfrac12$; and, over \emph{every} legal Step-2 assignment and
\emph{every} valid matching of those instances, that no execution reaches
Theorem~\ref{thm:belowceil}'s ceiling. Exact arithmetic is required for these checks because the construction sits on
boundary equalities: shares equal to $\beta$, a big-job budget equal to $1$, and
a rider of size exactly $T/2$.

The run additionally verifies five diagnostic conditions:
\begin{itemize}
\item the same instances at $\beta = 2/3$ reach $1.5\,T$, below
      $\tfrac{11}{6}$;
\item the threshold is pinned by a job of size exactly $T$;
\item Theorem~\ref{thm:ratio}'s chain, rebuilt at target $T = 1 + 2\delta$ with
      forced singletons saturating the partners in place of this note's pin,
      reaches a makespan of exactly $\tfrac{11}{6}$ at $\delta = 1/100$ through
      this file's independent slot-and-matching code. That is a ratio of
      $275/153 \approx 1.7974$ to \emph{that} target, strictly inside
      Proposition~\ref{prop:ceiling}'s bound of $\tfrac{11}{6}T = 1.8700$. It is
      not $\tfrac{11}{6} - \delta$, because this control is a
      scaled variant of the family rather than the family, and its singletons
      inflate its optimum to $163/120$ --- which is why the optimum-relative
      families of Theorem~\ref{thm:optexact} replace those singletons with a
      pin;
\item two big jobs stack on one machine at $\beta = \tfrac12$ and cannot just
      above it;
\item and \texttt{p29amb\_ws\_algorithm}'s floating-point implementation, handed
      the same $x$, reproduces all seven family makespans.
\end{itemize} No random search was run below one half with stacking enabled, so
the nested family is a construction rather than a search result.
\end{measurement}

\begin{measurement}[the nested family, and the closed form, checked
exactly]\label{meas:belowexact}
Two further runs, both in exact rational arithmetic.
\texttt{p29sup\_beta\_below\_half.py} builds
Theorem~\ref{thm:belowexact}'s nested family at seven thresholds and confirms the
approach rather than assuming it: the gap to
$\tfrac32 + (1-\beta)\lfloor 1/\beta\rfloor$ falls by a factor of ten for each
decade of $\varepsilon$ at $\beta = \tfrac25$ and $\beta = \tfrac18$, which is the
linear convergence the proof claims and not a plateau.
\texttt{p29sup\_N2\_nonreciprocal\_verify.py} then instantiates
Remark~\ref{rem:midpoint}'s closed form at $140$ non-reciprocal thresholds: every
non-reciprocal reduced fraction in $(0, \tfrac12]$ with denominator at most $30$, plus three points $1/1000$,
$1/5000$ and $1/25000$ below each reciprocal $1/k$ for $k = 2, \dots, 11$ --- the
tightest cases, where the slack $g = 1 - k\beta$ and the margin $g/2$ are
smallest --- and checks at every one of them that
the instance is feasible with $T_{\mathrm{LP}} = \OPT = T$, that the Step-2 branch
is legal, that the collector takes exactly two slots, and that the worst valid
matching finishes at $(k + \tfrac12 + \tfrac{g}{2})T$, strictly above
$\lfloor 1/\beta\rfloor + \tfrac12$ and strictly below the ceiling. The smallest
margin in the sweep is at $\beta = 12499/25000$, where $g/2 = 1/25000$.

The makespan is computed twice at each threshold --- once from the closed form and
once by walking every valid matching --- and the reciprocal thresholds are the
negative control: there $g = 0$, the same construction finishes \emph{below}
$k + \tfrac12$, and the family gains nothing, as it must, since the budget is
already spent by the seized jobs. Outside the sweep: irrational
thresholds, which only the proof reaches, since rational arithmetic cannot
instantiate them; and denominators above $30$ other than the near-reciprocal
points added explicitly.
\end{measurement}

\medskip
\noindent
\emph{The sweep that came first.} Before any of the above, the failure of the
$\tfrac{11}{6}$ ceiling below one half was recorded as a measurement, and it is
kept because it is evidence of a different kind --- what happens on random
instances rather than on a construction --- and because its own honest limit is
what made the construction findable. It excluded every Step-2 branch that puts
two big jobs on one machine, on the note's argument, which is valid only above
$\beta = \tfrac12$; that is exactly the branch Theorem~\ref{thm:belowfamily}
uses. It voids the step Proposition~\ref{prop:ceiling}'s proof leans on to
fix the base loads. The relaxation's big-job budget of $1$ holds at every
threshold; what needs $\beta > \tfrac12$ is that a machine receives at most one
\emph{seized} big job in Step~2, since two seized shares of at least $\beta$
then exceed the budget. At $\beta = \tfrac18$ up to eight shares of $\tfrac18$
fit within it, and every job seized at such a share lands on the machine in
full.

\begin{measurement}[the ceiling is a fact about $\beta = 2/3$]\label{meas:ceilingoffbeta}
Below $\beta = 1/2$, some instance, some feasible $x$, some legal Step-2
branch and some valid matching finish \emph{above} $\tfrac{11}{6}T$.

Our exact-rational witness has seven machines and thirteen jobs at
$\beta = 1/8$, with sizes and allowed pairs
\[
\begin{aligned}
\bigl[\,&(1,\{5,6\}),\ (1,\{2,6\}),\ (1,\{2,6\}),\ (0.512,\{5,6\}),\ (0.462,\{1,4\}),\\
        &(0.516,\{1,6\}),\ (0.533,\{0,1\}),\ (0.541,\{3,4\}),\ (0.497,\{2,3\}),\\
        &(1,\{0,2\}),\ (0.465,\{2,3\}),\ (0.502,\{3,4\}),\ (1,\{1,4\})\,\bigr].
\end{aligned}
\]
Its least feasible threshold is exactly $T = \tfrac{2257}{1750} \approx 1.289714$,
the total size $\tfrac{2257}{250}$ spread evenly over the seven machines. Two facts
pin it. Spreading evenly is necessary: the seven load rows add up to the total
size, so no smaller target is feasible. It is also sufficient, because one
fractional solution meets every row at this target. Put on the first listed
machine of each job, in the order displayed, the share
\[
\bigl(1,\ 1,\ 0,\ \tfrac{507}{896},\ 1,\ \tfrac{785}{903},\ \tfrac{156}{287},\ 1,\
\tfrac{2028}{3479},\ 1,\ 0,\ \tfrac{535}{3514},\ \tfrac{951}{7000}\bigr),
\]
and the rest on the second. Every machine is then loaded to exactly $T$, and no
machine carries more than one big job's worth of fraction; this point is a
vertex of the relaxation. Five jobs exceed $T/2$, the five of size $1$. At
$\beta = \tfrac18$ a legal Step-2 branch sends them, in display order, to
machines $5, 2, 6, 0$ and $1$, and a valid matching then finishes at exactly
$\tfrac{2511}{1000}$, which is $\tfrac{17577}{9028}\,T \approx 1.946943\,T$, against
$\tfrac{11}{6} \approx 1.833333$. No execution at this point does worse.
Measured against the optimum the same run is milder: $\OPT = \tfrac{189}{125}$,
so it is $\tfrac{93}{56} \approx 1.66$ times optimal. The witness breaks the bound
Proposition~\ref{prop:ceiling} proves, which is relative to $T$; against the
optimum, the guarantee below one half is the subject of Section~\ref{sec:below}.
The deposited \texttt{p29\_extensions/\allowbreak p29sup\_witness7\_exact.py}
checks the point, its vertex rank, the branch, the matching and the optimum in
rational arithmetic.

\emph{The precision of these numbers.} The sizes above are terminating decimals
and are read as the rationals they name ($0.512$ is $512/1000$). The
seven-machine witness is exact: its threshold, its fractional point, its Step-2
branch, its matching and its optimum are checked in rational arithmetic. The
six-machine maximum below is not. Unlike every other measurement in this note it
is a \emph{floating-point} result: $T$ is a bisection to $10^{-7}$ over a
linear-programming feasibility test accurate to about $10^{-6}$, and the
deposit holds no second implementation that recomputes it, so neither a coding
error nor a conditioning one is excluded. Nothing proved in this note
rests on either witness: their job is to show that the $\tfrac{11}{6}$ ceiling
is a fact about $\beta = 2/3$ and not about the scheme, and the exact statements below one half are the theorems of
Section~\ref{sec:below}, certified in exact rational arithmetic by
Measurements~\ref{meas:belowhalf}, \ref{meas:belowexact} and \ref{meas:repairs}.

A sweep of $4{,}000$ random instances at
$\beta \in \{\tfrac12, \tfrac25, \tfrac13, \tfrac14, \tfrac15, \tfrac18\}$
(seed $77$) reaches $2.021834\,T$, on six machines and eleven jobs at
$\beta = 1/8$. Counting the instances on which some legal branch finishes
\emph{strictly} above $\tfrac{11}{6}T$ --- none merely reached it --- gives

\smallskip
\centerline{\begin{tabular}{lcccccc}
\hline
$\beta$ & $1/2$ & $2/5$ & $1/3$ & $1/4$ & $1/5$ & $1/8$\\
instances above $11/6$ & $\mathbf{0}$ & $5$ & $8$ & $10$ & $16$ & $27$\\
\hline
\end{tabular}}
\smallskip

\noindent
No sampled instance exceeds $\tfrac{11}{6}T$ at $\beta = 1/2$, but that zero
describes the restricted sweep and not the threshold: the sweep omits the
Step-2 branches that stack big jobs, and at $\beta = \tfrac12$ the exact
guarantee is $\tfrac52$ under the weak test (Theorem~\ref{thm:belowexact}) and
$2$ under the strict one (Proposition~\ref{prop:strictbelow}), both above
$\tfrac{11}{6}$. The larger counts further below one half describe the sample,
not a monotone law.

Limitations of the sweep. The maximum is a lower bound on a supremum, because the
search \emph{excluded} every Step-2 branch that assigns two big jobs to one
machine --- on the note's own argument, which is valid only above
$\beta = 1/2$ and so excludes exactly the branches that may be legal here.
And $4{,}000$ instances at six thresholds is a sample, not a range: what is
checked is each named witness, the seven-machine one exactly and the
six-machine maximum in floating point, not a statement
about all $\beta < 1/2$. \texttt{p29amb\_ceiling\_under\_free\_step2.py}
regenerates the sweep and the counts.
\end{measurement}

\begin{figure}[!b]
\centering
\includegraphics[width=0.97\textwidth]{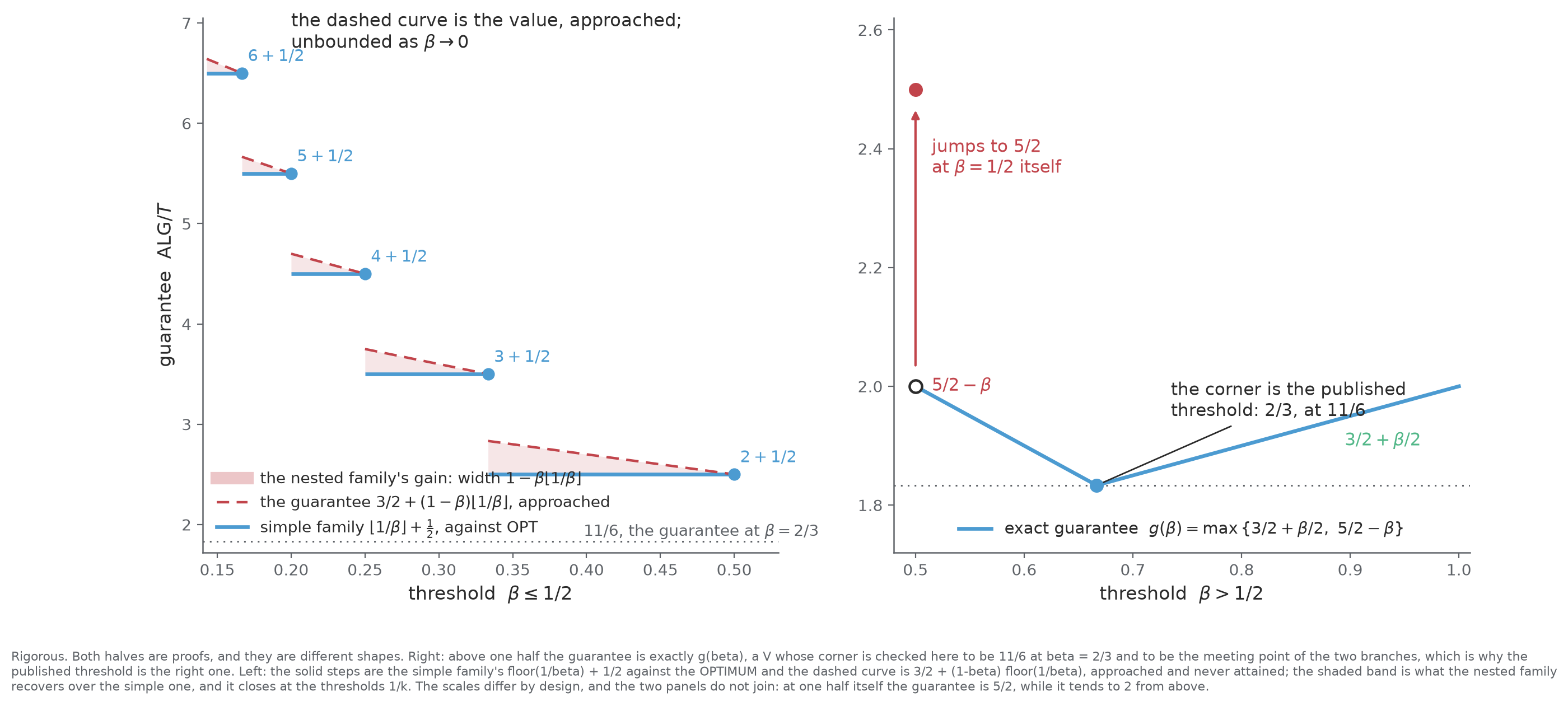}
\caption{The guarantee at every threshold, both halves proved, on two scales
because the regions have different ones. \emph{Right:} above one half the
guarantee is exactly $g(\beta)$ --- Proposition~\ref{prop:beta}'s families from
below, Proposition~\ref{prop:betaceil}'s ceiling from above
(Corollary~\ref{cor:betaexact}) --- a V whose corner is $\beta = 2/3$ at exactly
$11/6$. The corner is the content: it is why the published threshold is the
right one, and every other threshold in the interval is strictly worse.
\emph{Left:} at or below one half the dashed curve
$\tfrac32 + (1-\beta)\lfloor 1/\beta\rfloor$ \emph{is} the guarantee, approached
and never attained (Theorem~\ref{thm:belowexact}): Theorem~\ref{thm:belowceil}
caps it there and the nested family reaches it in the limit. The solid curve is
what the simpler family of Theorem~\ref{thm:belowfamily} reaches,
$\lfloor 1/\beta\rfloor + \tfrac12$, also against the optimum, and the shaded
band is what spending the collector's remaining load budget on further slots
recovers. The band closes at the thresholds $1/k$, marked with dots, where the
two constructions already agree. The two panels deliberately do not join: at one half itself the
guarantee is $\tfrac52$ while it tends to $2$ from above, because the count
$\lfloor 1/\beta\rfloor$ passes from one to two there. Both sides are proved; the random-instance sweep of
Measurement~\ref{meas:ceilingoffbeta} is not drawn here.}
\label{fig:threshold}
\end{figure}

\section{Summary of results}\label{sec:summary}

We conclude by collecting the analytic results, the computational
observations, and the questions that remain open. Every line identifies the result supporting it.

\paragraph{Proved.}
Together the two freedoms in Wang and Sitters' algorithm permit exactly
$\tfrac{11}{6}$, and no more: Theorem~\ref{thm:ratio} exhibits instances, feasible solutions
and valid matchings reaching $\tfrac{11}{6} - \delta$ for every
$\delta \in (0, 1/6)$, and Proposition~\ref{prop:ceiling} keeps every machine
\emph{strictly} below $\tfrac{11}{6}T$ on every instance, so the supremum is
$\tfrac{11}{6}$ and no instance attains it. A second, six-job family
(Theorem~\ref{thm:tight}) is tight only against the relaxation's target, where
its ratio to the true optimum tends to $\tfrac{11}{8}$; the difference between those
two scales is the distinction the note keeps throughout.

Changing Wang and Sitters' threshold cannot help, and the range is now covered
end to end. Above one half the guarantee against the relaxation's own value is exactly $g(\beta) = \max\{\tfrac32 + \tfrac\beta2,\ \tfrac52 - \beta\}$: Proposition~\ref{prop:beta} establishes the lower bound with two families each covering an interval of thresholds rather than a grid, and Proposition~\ref{prop:betaceil} establishes the upper bound by taking Proposition~\ref{prop:ceiling}'s three cases with $\beta$ in them (Corollary~\ref{cor:betaexact}). At or below one half every big job is seized in
Step~2 (Lemma~\ref{lem:allseized}), a machine can be handed
$\lfloor 1/\beta\rfloor$ of them, and Theorem~\ref{thm:belowfamily} builds
instances with $T_{\mathrm{LP}} = \OPT$ where exactly that happens: the
guarantee is at least $\lfloor 1/\beta\rfloor + \tfrac12$ \emph{against the
optimum} and at most $\tfrac32 + (1-\beta)\lfloor 1/\beta\rfloor$
(Theorem~\ref{thm:belowceil}); Theorem~\ref{thm:belowexact} closes the interval
between the two at \emph{every} threshold below one half, not only at the
thresholds $1/k$ where they already agreed, and the exact value is
$\tfrac32 + (1-\beta)\lfloor 1/\beta\rfloor$, approached and never attained. The
same question above one half against the \emph{optimum} rather than the
relaxation's value is settled too: Theorem~\ref{thm:optexact} replaces
Proposition~\ref{prop:beta}'s forced singletons, which inflate the optimum, with
a pin that does not, and gets $g(\beta)$ there as well. So $\tfrac23$ is the
unique minimizer over all of $(0,1)$, the guarantee is unbounded as the
threshold falls, and it jumps at one half rather than degrading smoothly
(Corollary~\ref{cor:wholerange}). Corollary~\ref{cor:optwholerange} states both
branches as one function of $\beta$ against the optimum.

\paragraph{Measured, and not proved.}
Proposition~\ref{prop:interp}'s two branches are attained only in the limit,
among other points at their own corners, as its restatement check records.
Measurement~\ref{meas:ceilingoffbeta} certifies its seven-machine witness below
$\beta = \tfrac12$ in exact rational arithmetic and recomputes its six-machine
maximum in floating point, but $4{,}000$ instances at six
thresholds is a sample and not a statement about all $\beta < \tfrac12$; what is
now a statement about all of them is Theorem~\ref{thm:belowceil}, proved rather
than sampled, and the sweep is corroboration.
Measurement~\ref{meas:belowhalf} is arithmetic verification of
Theorem~\ref{thm:belowfamily}'s family at seven thresholds, not evidence for the
theorem, which is proved at every threshold. Every
statement labelled \emph{Measurement} reports what the deposited code returned
at the stated parameters; none of them is a proof.

\paragraph{Open.}
The threshold-dependent guarantee is determined in both directions by
Corollary~\ref{cor:optwholerange}, so what remains is elsewhere. Below one half
the interval between
family and ceiling, of width $1 - \beta\lfloor 1/\beta\rfloor$ at thresholds
that are not reciprocals of whole numbers, is closed by
Theorem~\ref{thm:belowexact}; above one half the same question against the
\emph{optimum} rather than the relaxation's value is closed by
Theorem~\ref{thm:optexact}. The questions the companion note \cite{Paper2} closes --- the
oracle, the vertex restriction, and both $7/4$ statements --- are not open
either. The open questions lie outside the threshold: whether a different
relaxation, a rounding rule that is not Shmoys--Tardos slot matching, or a Step-1
rule that chooses its feasible solution, which \cite{Paper2} leaves open, moves
the constant at all. Neither possibility is settled here.

Nothing in this note improves the approximation ratio for graph balancing,
which stands at $1.75$.

\appendix

\section{Restatement checks}\label{app:restatement}

A result's statement and proof can disagree:
a proof can quietly assume more than the statement grants, or establish less than
the statement claims. The checks below re-read one result against its own proof and
ask whether the quantifiers match --- what is asserted for \emph{every} object,
what is only exhibited for \emph{one}, and which hypotheses the proof actually uses.
They are addressed to a reader verifying the note rather than to one following its
argument, which is why they are here and not in the body.

\subsection*{Lemma~\ref{lem:chain}}

Lemma~\ref{lem:chain} is asserted for every
machine with at least one slot job and every valid matching, with no
hypothesis on $x$ beyond relaxation feasibility; the proof uses only the sort
order, the slot lengths, and the one-job-per-slot cap, so the quantifiers
match.

\subsection*{Theorem~\ref{thm:ratio}}

The claim was: for every $\delta$ in the stated
range there exist an instance, a feasible relaxation solution and a valid
matching with $\ALG/\OPT = 11/6 - \delta$. The proof exhibits all three
explicitly and verifies feasibility, optimality of the claimed $\OPT$,
the Step-2 dodge and the validity of the matching. The supremum statement
follows because $\delta$ ranges over an interval with infimum $0$. The
proof does \emph{not} claim that every run of the scheme on this
instance is bad, and it does not claim the bad point is a vertex of the
polytope; see the remark below.

\subsection*{Proposition~\ref{prop:ceiling}}

The claim quantified over all instances, all
$T$ at which the relaxation is feasible, all feasible $x$, and all valid
matchings, asserting a strict inequality per machine. The proof fixes an
arbitrary such tuple and an arbitrary machine, splits into three trivially exhaustive
cases, uses \cite{WS16}'s observation inside case (S), and derives a strict
inequality in each. The strictness in (S) comes from
$w_{K_i} > 0$; in (N0) from $3/2 < 11/6$; and in (N1) from
$x_q > \tfrac13$ alone, which already gives load $\le 2 - x_q/2 < \tfrac{11}{6}$
without any appeal to $w_{K_i}$. The intermediate strict step in that chain
does use $w_{K_i} > 0$; the conclusion does not. In (N0) the slot-$1$ term is
the size of the job \emph{matched} to slot $1$, not $s_1$; the two differ,
and $s_1$ can exceed $1/2$ there: machine $A_1$ of the two-phase instance of \cite{Paper2} has $s_1 = 1$ under a valid matching that
puts a job of size $1/2$ in its only slot. Cases (N0) and (S) are, up to
notation, the two cases of \cite[Lemma~11]{SY21}---their first case is our
(N0) and (N1) together, and gives the weaker bound $3/2 + \alpha/2$---%
including the tracking
of the Step-2 job's own size that sharpens \cite{WS16}'s flat
$5/2 - \beta$ to $3/2 + b_i/3$; the case analysis here is new only in
(N1), where carrying $x_q$ through gives the parametrised bound
$2 - x_q/2$.

\subsection*{Proposition~\ref{prop:interp}}

The bound is over \emph{every} valid matching,
and it is legitimate to take a maximum because $b_{\max}$, $p_{\max}$ and
$\sigma$ depend only on the instance, $x$ and Step~2 --- never on the
matching. Two things it does not say. Both branches are attained in
the limit, each at its own corner among other points: the second at
$(b_{\max}, p_{\max}, \sigma) = (0, 1, \tfrac12)$, by
Theorem~\ref{thm:ratio}'s family, where the first branch reads
$\tfrac32$; and the first at $(1, \tfrac12, \tfrac12)$, by
the threshold-relative family of \cite{Paper2}, where the second reads $\tfrac32$. Nothing here
shows the bound is sharp anywhere between those two corners. Its instance-level
corollary --- every Step-2 job at most $3/4$ and every job at most $7/8$
implies $\tfrac74$, since then $2p_{\max} + \sigma \le \tfrac74 +
\tfrac12$ --- still depends on the run, because which jobs Step~2 seizes depends
on $x$; its stronger form, every job at most $3/4$, is checked from the input
alone. And $\sigma$ itself depends on $x$ and on Step~2, so the refinement is at
the level of a certificate for a particular run rather than of a class of
instances. And $7/8$ is exactly the value of $p_{\max}$ at which the
$\sigma$ constraint stops binding, which is why that constant appears.

\subsection*{Theorem~\ref{thm:tight}}

The claim asserts relaxation feasibility exactly at
$T$, the $O(\delta)$-box, existence of a valid matching of makespan
$\tfrac{11}{6}T - O(\delta)$, the exact value of $\OPT$, and the
resulting limit $11/8$ for $\ALG/\OPT$. Each is established above. The
theorem does \emph{not} claim tightness of the approximation ratio; that
is Theorem~\ref{thm:ratio}.

\subsection*{Proposition~\ref{prop:beta}}

The proposition is about the guarantee against
$T$, not against $\OPT$, because that is what \cite{SY21}'s instances
witness: the distinction of Section~\ref{sec:sy}, and the reason their
Figure-4 instances sit at $1.37$ against the optimum while binding at $1.83$
against $T$. So changing $\beta$ cannot improve $11/6$ against the relaxation's value.
Proposition~\ref{prop:beta} itself is only $T_{\mathrm{LP}}$-relative;
\textbf{Theorem~\ref{thm:optexact} later proves the optimum-relative
analogue} above one half, and Theorem~\ref{thm:belowexact} gives it below.
Neither of \emph{these two} families binds against $\OPT$ below $2/3$, both
collapsing there for the reason in Remark~\ref{rem:beta}; the families that do
are Theorem~\ref{thm:optexact}'s, and they leave no threshold that improves the
approximation ratio.

Two quantifiers, both statements about an interval and both proved on it. The
construction is uniform in $\beta \in (\tfrac12, 1)$: the two families are
parameterized by $\beta$, and every condition their proofs need is a comparison
between affine functions of $\beta$, so each covers an interval of thresholds
rather than a point. The second quantifier is the minimizer, and an argmin over
an interval needs a strict inequality at every point of it:
$R_T(\beta) \ge g(\beta) > \tfrac{11}{6} = R_T(\tfrac23)$ for every
$\beta \ne \tfrac23$, so the argmin is $\{\tfrac23\}$. In the proposition $g$ is
a floor; the matching upper bound is Proposition~\ref{prop:betaceil}, stated and
proved separately.

\subsection*{Propositions~\ref{prop:betaceil} and Corollary~\ref{cor:betaexact}}

The proposition quantifies over all instances, all $T$ at which the relaxation
is feasible, all feasible $x$ and all valid matchings, at a fixed
$\beta \in (\tfrac12, 1)$, and asserts a strict per-machine inequality; the
proof fixes an arbitrary such tuple and an arbitrary machine and derives one in
each of three cases, so the quantifiers match. Its cases are
Proposition~\ref{prop:ceiling}'s, and it inherits their hypotheses: the
half-open slot convention of Section~\ref{sec:prelim}, and the exhaustiveness
argument that a machine holding a Step-2 job carries no unassigned big
fraction, which is re-derived at general $\beta$ rather than quoted, because the
published version of that step is stated at $\tfrac23$. The corollary's
equality is a supremum, not a maximum: the ceiling is strict on every execution
and the families approach $g(\beta)$ as their perturbation vanishes, so no
execution attains it. Nothing here is claimed at $\beta \le \tfrac12$, where
Proposition~\ref{prop:ceiling}'s case split is not exhaustive in the same way;
that region is Section~\ref{sec:below}.

\subsection*{Lemma~\ref{lem:allseized} and Theorem~\ref{thm:belowceil}}

The lemma asserts that at $\beta \le \tfrac12$ every big job is seized. It uses
only that a job's fraction is supported on at most two machines and sums to one,
which is the relaxation's first constraint together with the problem's own
restriction; instances outside that model are outside the lemma. The theorem
quantifies over all instances, all feasible $x$, all \emph{legal Step-2
assignments} and all valid matchings. The third of these reflects the enlarged
execution space at $\beta \le \tfrac12$, which is why
Section~\ref{sec:below} extends the definition before stating anything. The
proof applies Lemma~\ref{lem:chain}(ii) with the Step-2 \emph{total} in place of
a single Step-2 job's size, which the lemma permits: its proof adds seized sizes
to matched sizes and never counts the seized jobs. The bound is on
$\ALG/T$; against the optimum it is at least as strong, since
$T_{\mathrm{LP}} \le \OPT$.

\subsection*{Theorem~\ref{thm:belowfamily} and Corollary~\ref{cor:wholerange}}

The theorem is an existence claim, and exhibits one instance, one feasible
solution, one legal Step-2 assignment and one valid matching per $\beta$; it does
\emph{not} claim that every execution on those instances is bad, and the
verification in Measurement~\ref{meas:belowhalf} reports the maximum over all of
them rather than assuming it. Two hypotheses do real work and are stated where
they are used: the rider has size exactly $T/2$ and so is not big, the
definition of big being strict; and the threshold is pinned at $T$ by a job of
size exactly $T$, without which the ratio would be against a target chosen
rather than forced. The value is attained when $\beta\lfloor 1/\beta\rfloor < 1$
and only approached otherwise, and the statement separates the two cases instead
of reporting a supremum as a maximum. The corollary's part~(iii) is an equality
only at reciprocal thresholds; parts~(i) and~(ii) are inequalities, and the
minimizer claim over $(0,1)$ combines them with
Corollary~\ref{cor:betaexact}.
\section*{Code and data}
The code and artifacts behind the computational results reported in this note
are deposited on the Open Science Framework as
\texttt{osf\_p29sup.zip} at
\url{https://doi.org/10.17605/OSF.IO/SKX86}, the same project that holds
the deposit of \cite{Shavit26}: the claim-supporting scripts, the modules they
import, the artifacts they write, a driver that runs them, and pinned
dependency versions. The deposit is the replication set, not the
project's working directory. Its README maps each claim to the command and
artifact that establish it and lists every file it holds.

The deposit reproduces every numerical measurement reported here with two
exceptions, each recorded where it occurs: the GLOP run inside
Measurement~\ref{meas:eightmatchings}, which no deposited file regenerates, and
the plateau-walk instance of Measurement~\ref{meas:corroborate}(a), which the
deposit ships and a rerun replaces (the ratio reproduces, the instance does
not). One further exception is recorded in the README: the script that draws
Figure~\ref{fig:threshold} is not deposited, since that figure plots closed
forms proved in the text rather than anything measured.

All other computational results reported here were regenerated from the
deposited code and artifacts. The public package contains the final claim-supporting
checks; exploratory proof attempts and the internal revision diary are not part
of the replication set.

\section*{Acknowledgements}

I thank Baruch Schieber (NJIT) for introducing me to the
restricted-assignment and graph-balancing gap problems; this note began
from reading Wang and Sitters within the problem area he opened up for
me. I thank JMS and CS for reading and commenting on drafts.

\section*{Use of AI tools}

AI tools assisted under the author's direction. Claude (Anthropic) wrote search,
verification and figure code to the author's specification. Claude and
ChatGPT (OpenAI) gave editorial support. Perplexity was used for literature search and citation verification.
ChatGPT (OpenAI) and Gemini (Google) were used for adversarial review and
independent mathematical checking.
The three-job family of Theorem~\ref{thm:ratio} first emerged during
an adversarial review run by Claude (Anthropic) and was then checked by the deposited
verification and the proof given here.

\section*{Author's note}

I offer this work in the spirit of community-supported research---the
distributed protein-folding projects are the model I have in mind---where the
computing is contributed rather than bought, the claim-bearing record is
public, and the worth of the result is what others can build on it. The public
replication package maps every computational result reported here, including
reported searches that found nothing, to the code that produces it and to a
stored artifact or log where the run creates one. It is a curated replication
set, not the entire development record. Additional retained developmental
materials, including superseded approaches, refuted attempts, and unreported
exploratory searches, are available upon reasonable request, subject to privacy
and copyright constraints. The negative results are part of the evidence here.
I hope the work is of use to the field.

\vfill
\begin{center}
\begin{minipage}{0.72\textwidth}
\itshape\small
3 empty slots\\
$1$, $1/2$, $1/3$\\
$11/6$ for the horizon,\\
aspire always\ldots{} never match
\end{minipage}
\end{center}

\end{document}